%% file: main-arXiv.tex
\documentclass[journal]{IEEEtran}
\usepackage{cite}
\usepackage{amsmath,amssymb,amsfonts}
\usepackage{graphicx}
\usepackage{textcomp}
\usepackage[table]{xcolor}
\usepackage{subfigure}
\usepackage{threeparttable}
\usepackage{amssymb,mathrsfs,amsmath}
\usepackage{url}
\usepackage{makecell}
\usepackage{mdwtab}
\usepackage{booktabs}
\usepackage{array}
\usepackage{mdwmath}
\usepackage{eqparbox}
\usepackage{multirow}
\usepackage{color}
\usepackage{threeparttable}
\usepackage{array}
\usepackage{booktabs}
\usepackage{textcase}
\usepackage{siunitx}
\usepackage{fbox}

\usepackage{booktabs}

\usepackage{pifont} 

\usepackage{bbding} 
\newcommand{\cmark}{\textcolor{green!80!black}{\ding{51}}}
\newcommand{\xmark}{\textcolor{purple}{\ding{55}}}

\newenvironment{packeditemize}{
	\begin{list}{$\bullet$}{
			\setlength{\labelwidth}{4pt}
			\setlength{\itemsep}{0pt}
			\setlength{\leftmargin}{\labelwidth}
			\addtolength{\leftmargin}{\labelsep}
			\setlength{\parindent}{0pt}
			\setlength{\listparindent}{\parindent}
			\setlength{\parsep}{0pt}
			\setlength{\topsep}{1pt}}}{\end{list}}

\definecolor{light_cyan}{rgb}{0.53, 0.75, 0.77}
\definecolor{light_blue}{rgb}{0.466, 0.655, 0.94}
\definecolor{light_pink}{rgb}{0.98, 0.55, 0.565}
\definecolor{light_yellow}{rgb}{0.98, 0.83, 0.32}

\renewcommand*{\arraystretch}{1.5}%
\definecolor{tabred}{RGB}{230,36,0}%
\definecolor{tabgreen}{RGB}{0,116,21}%
\definecolor{taborange}{RGB}{250,124,30}%
\definecolor{tabbrown}{RGB}{171,70,0}%
\definecolor{tabyellow}{RGB}{251,253,169}%
\newcommand*{\vcorr}{%
  \vadjust{\vspace{-\dp\csname @arstrutbox\endcsname}}%
  \global\let\vcorr\relax
}%

\usepackage{algorithm}
\usepackage{diagbox}
\usepackage[section]{placeins}
\usepackage{algpseudocode}
\usepackage[inkscapelatex=false]{svg}

\usepackage{amsthm}
\theoremstyle{plain}       
\newtheorem{thm}{Theorem}

\theoremstyle{definition}

\newtheorem*{prf}{Proof}

\usepackage{url}

\usepackage{hyperref}
\hypersetup{
     colorlinks = true,
     linkcolor = blue,
     anchorcolor = red,
     citecolor = magenta,
     filecolor = blue,
     urlcolor = black,
}

\algnewcommand{\LeftComment}[1]{\Statex \(\triangleright\) #1}

\def\BibTeX{{\rm B\kern-.05em{\sc i\kern-.025em b}\kern-.08em
    T\kern-.1667em\lower.7ex\hbox{E}\kern-.125emX}}

\begin{document}

\title{Client-Cooperative Split Learning }

\author{\IEEEauthorblockN{Haiyu Deng\IEEEauthorrefmark{1}, Yanna Jiang\IEEEauthorrefmark{1}, Guangsheng Yu\IEEEauthorrefmark{1}, Qin Wang\IEEEauthorrefmark{1}\IEEEauthorrefmark{2},  \\
Xu Wang\IEEEauthorrefmark{1}, Wei Ni\IEEEauthorrefmark{3}, Shiping Chen\IEEEauthorrefmark{2}, Ren Ping Liu\IEEEauthorrefmark{1}
}

\IEEEauthorrefmark{1}\textit{University of Technology Sydney} $|$
\IEEEauthorrefmark{2}\textit{CSIRO Data61}  $|$
\IEEEauthorrefmark{3}\textit{Edith Cowan University}
\thanks{H. Deng and Y. Jiang contributed equally.}
}

\maketitle

\input{sections_1/0_abstract}

\begin{IEEEkeywords}
Split Learning, Privacy, Copyright Protection 
\end{IEEEkeywords}


\input{sections_1/1_introduction}
\input{sections_1/2_warmups}
\input{sections_1/3_system}
\input{sections_1/4_design}

\input{sections_1/5_experiment}
\input{sections_1/6_relatedwork}

\input{sections_1/7_conclusion}

\bibliographystyle{bibliography/IEEEtran}
\bibliography{bibliography/bib}

\end{document}

%% file: sections_1/0_abstract.tex
\begin{abstract}

\textcolor{black}{Model training is increasingly offered as a service for resource-constrained data owners to build customized models. Split Learning (SL) enables such services by offloading training computation under privacy constraints, and evolves toward \textit{serverless} and \textit{multi-client} settings where model segments are distributed across training clients.}
This cooperative mode assumes partial trust: data owners hide labels and data from trainer clients, while trainer clients produce verifiable training artifacts and ownership proofs.
We present \textsc{CliCooper}, a multi-\underline{cli}ent \underline{cooper}ative SL framework tailored for cooperative model training services in heterogeneous and partially trusted environments, where one client contributes data, while others collectively act as SL trainers. \textsc{CliCooper} bridges the privacy and trust gaps through two new designs. First, differential privacy-based activation protection and secret label obfuscation safeguard data owners’ privacy without degrading model performance. Second, a dynamic chained watermarking scheme cryptographically links training stages on model segments across trainers, ensuring verifiable training integrity, robust model provenance, and copyright protection. 
Experiments show that \textsc{CliCooper} preserves model accuracy while enhancing resilience to privacy and ownership attacks. It reduces the success rate of clustering attacks (which infer label groups from intermediate activation) to 0\%, decreases inversion-reconstruction (which 
recovers training data) similarity from 0.50 to 0.03, and limits model-extraction–based surrogates to about 1\% accuracy, comparable to random guessing.

\end{abstract}

%% file: sections_1/1_introduction.tex
\section{Introduction}\label{sec:intro}

Under privacy and data-sovereignty constraints, Split Learning (SL) is a promising technique that enables Artificial Intelligence (AI) model training services for resource-constrained data owners without exposing raw data: clients compute intermediate activations, and a server completes training~\cite{twoPartySL_5,matsubara2022split_6,outsource_4,hu2025split_7}. This paradigm assumes a well-provisioned server fully trusted for computation and provenance~\cite{jung2025split}. In practice, deployments differ: personal and edge devices offer surplus yet fragmented compute, which is insufficient individually for large models but becomes valuable when aggregated~\cite{zafari2020let}.
This motivates a \textbf{serverless, multi-client collaboration}~\cite{lin2024split, rawal2013multi}: one client provides data but limited compute; several other clients contribute compute and jointly play the ``server'' role. 
This emerging setting becomes both feasible and desirable, better meeting costs by amortizing idle edge/GPU cycles~\cite{li2024split}
, avoiding centralized rental, and improving availability by enabling elastic, on-demand participation~\cite{hammoud2022demand}.

However, this multi-client collaboration setting is only partially trusted~\cite{twoPartySL_5}.
This shift raises coupled challenges that are insufficiently addressed by current SL frameworks:
\begin{packeditemize}

\item \textcolor{violet}{\textbf{RQ1} (data privacy)}: \textit{How can the data-providing client preserve the privacy of raw inputs and labels while still enabling effective SL?}

\item \textcolor{violet}{\textbf{RQ2} (layer ownership)}: \textit{How can trainer clients present cryptographically anchored ownership claims for their contributions to access compensation?}

\item \textcolor{violet}{\textbf{RQ3} (unauthorized-use defense)}: \textit{How can the collaboratively trained model resist unauthorized use?}

\end{packeditemize}
Building on these RQs, we present \textsc{CliCooper}, a multi-client cooperative SL framework tailored to the serverless, partially trusted, multi-client setting. 
\textsc{CliCooper} integrates secret mapping labels with Differential Privacy (DP)-guarded activations and a provenance-preserving watermark chain.

For the data-providing client, raw data does not need to be disclosed; instead, each true class is mapped to a set of pseudo-labels, which hides task semantics and label quantity while expanding the effective label space.
In parallel, data are augmented to match the expanded label space, and inputs are passed through an embedder whose intermediate activations are perturbed with calibrated DP noise before any upload.
Training thus proceeds only on DP-protected activations paired with pseudo-labels, which (i) suppresses raw-input and label inference to address \textbf{\textcolor{violet}{RQ1}}, and (ii) binds the model’s practical usefulness to authorization, since only parties holding the secret true-to-pseudo mapping can faithfully interpret predictions, thereby constraining unauthorized reuse for \textbf{\textcolor{violet}{RQ3}}.

For the trainer clients, after coordinating the execution order and model structure, each receives only the upstream activation, keeps its local parameters opaque, and forwards its output activation to the next participant. 
At agreed checkpoints, a chained watermark is derived from the received activation digest and then embedded into the client’s segment.
The resulting watermarks are cryptographically linked across segments, creating a tamper-evident lineage that ties each segment to the collaborative run and the intended task. 
This lineage enables auditable ownership assertion, lawful entitlement to compensation, and accountability against segment theft or task-mismatch reuse, fulfilling \textbf{\textcolor{violet}{RQ2}} while preserving the confidentiality constraints of \textbf{\textcolor{violet}{RQ3}}.

The key contributions can be summarized as follows:

\begin{packeditemize}

\item We address privacy risks in partially trusted, client-cooperative SL by concealing label semantics and reducing exposure of raw inputs (via label expansion and calibrated noise on intermediate activation), preventing task leakage and curtailing inversion while keeping the true–pseudo mapping under the data owner’s control.

\item We ensure training traceability and model ownership in multi-client settings by binding checkpoints with cryptographically chained watermarks, enabling verifiable training integrity, robust provenance, and fair compensation, and deterring fabricated training trajectories.

\item We validate at scale across diverse datasets and architectures: baseline accuracy is preserved and improves by up to 2\%; internal clustering attacks on client-supplied activation drop to 0\%; inversion-reconstruction similarity falls from 0.50 to 0.03; and in black-box settings, model-extraction attacks fail, with surrogate models trained on API-labeled data achieving $\sim$1\% accuracy (random guessing).

\end{packeditemize}

The rest of this paper is organized as follows. \S\ref{sec:warm} presents state-of-the-art SL and watermarking technologies. \S\ref{sec:system} introduces the system model. \S\ref{sec:design} describes our framework. \S\ref{sec:experiment} shows empirical results. \S\ref{sec:related work} summarizes privacy threats and copyright protection. \S\ref{sec:conclusion} concludes this paper.

%% file: sections_1/2_warmups.tex
\section{Technical Warm-ups}\label{sec:warm}

\subsubsection{Split Learning Architectures}

SL allows multiple entities to collaboratively train a model while keeping local data private, typically in two ways: 
\begin{packeditemize}
    \item In \textit{vanilla SL}~\cite{hu2025split_7,vepakomma2018split_13,poirot2019split_40, jung2025split}, the client encodes local data and sends the resulting intermediate activation to the server, which performs most of the computation and, during gradient descent, requires access to label information~\cite{twoPartySL_5}. 
    \item \textit{U-shaped SL}~\cite{gupta2018distributed_12,u_shaped_38,u_shaped_39} is tailored to unmet label-privacy requirements in vanilla SL settings. A client will not share labels with the server.  By preprocessing activations after the split on the client side, the client preserves label confidentiality while maintaining the integrity.
\end{packeditemize}
Both assume a server with sufficient compute to train the model independently. In Vanilla SL, clients share labels for gradient computation, which exposes sensitive information. In contrast, U-shaped SL moves gradient computation to the client side, thereby preserving label privacy. However, this design increases communication overhead during the forward and backward passes and adds extra local computation. 

\subsubsection{Watermarking} 

Watermarking via modifying model parameters (weights $W$) was first proposed (EWDNN) by Uchida et al.~\cite{uchida2017embedding_18}. A binary watermark $\Lambda \in\{0,1\}^{B}$ is embedded during training by augmenting the objective:
    $loss(W) = l_w(W) + \lambda l_\Lambda(W)$,
where $l_w$ denotes the main task loss, $\lambda$ is an adjustable trade-off coefficient, and $l_\Lambda$ is the watermark-embedding regularizer:
    $l_\Lambda(W) = -\sum_{i=1}^{B} (\Lambda_{i}\log(\hat\Lambda_i)+(1-\Lambda_{i})\log(1-\hat\Lambda_{i}^{'}))$,
where $B$ denotes the watermark length and $\hat\Lambda$ is the extracted watermark. During embedding, it 
assesses the similarity between the ground-truth watermark and the extracted watermark. The extracted watermark is $\hat\Lambda^{'} = \Phi (\sum kW)$, where $k$ is the secret key used for watermark embedding, and $\Phi (z) = \frac{1}{1-\exp(-z)}$ is the step function. 

In general, weight-modified-based watermarking relies on the embedding regularizer. Wang et al.~\cite{Wang2019Attack_14} showed that such modifications distort the weight distribution, allowing adversaries to detect the presence of a watermark. RIGA~\cite{wang2021riga_19} redesigned the loss function to mitigate this vulnerability:
    $loss(W) = l_w(W) + \lambda l_\Lambda(W)+\beta l_\Delta (W)$,
where $ \beta$ is a tunable coefficient and $l_\Delta$ measures discrepancy between the clean and watermarked models, constraining their divergence to deter watermark detection via weight-distribution tests. FedIPR~\cite{li2022fedipr_20} applied model watermarking to federated learning, using a watermarking method aligned with EWDNN/RIGA.

%% file: sections_1/3_system.tex
\section{System Model}\label{sec:system}

We target a serverless, partially trusted, multi-client cooperative SL model with heterogeneous participants~\cite{IoT}: \textit{one data client owns the private dataset but has limited compute, while several trainer clients contribute surplus yet fragmentary compute}. The group jointly trains a split model by agreeing on split points and an execution order. Each trainer handles a local segment sized to its capacity.
Parties are honest but curious, and there may be divergent interests, so the design must (i) preserve raw input/label privacy for the data owner, (ii) support auditable ownership and lawful entitlement for trainers, and (iii) constrain unauthorized model reuse.

We integrate a privacy-and-authorization layer that couples secret-mapping label expansion with DP protection. The data client never discloses its true labels; instead, each true class is mapped to a set of pseudo-labels, obfuscating task semantics and label quantities. Data augmentation~\cite{data_augment} is introduced to align sample counts with the expanded label space.
Inputs are then passed through an embedder, and the resulting intermediate activations are perturbed with DP noise before any upload, mitigating feature inversion and label inference attacks~\cite{unsplit_27, inversion_26}. 
Training proceeds only on these DP-protected activations paired with pseudo-labels. 
For authorized users who hold the secret true$\leftrightarrow$pseudo mapping, model utility remains intact, while for unauthorized parties, outputs degrade to an unusable level.

\textsc{CliCooper} weaves a chained watermark through the collaboration process. 
At agreed checkpoints, each trainer derives a watermark token from a digest of the received activation and then embeds it into its segment. 
Tokens are cryptographically linked across segments to form a tamper-evident lineage that ties each segment to the collaborative run and intended task, supporting auditable ownership assertions, compensation eligibility, and accountability against segment theft or task-mismatch reuse.
Each trainer exposes only the upstream activation and keeps its local segment opaque to collaborators.
External adversaries interact with the final collaboratively trained model purely as a black box~\cite{backdoor_wm_33}. Under this visibility constraint, our chained watermarking design can preserve the confidentiality guarantees required in a partially trusted setting and is resilient to common attacks, including removal~\cite{kassis2025unmarker_24}, overwriting~\cite{zong2024ipremover_23}, and fine-tuning~\cite{pegoraro2024deepeclipse_21}.

\subsection{Entities} 
\textsc{CliCooper} involves three primary entities under a serverless, partially trusted, multi-client regime:
\begin{packeditemize}
    \item \textit{Data client $\mathcal{C}$} owns the private dataset with ground-truth labels, but has only lightweight computing adequate for embedding/perturbation rather than end-to-end training.
    Its primary requirement is to prevent raw input and label inference and hide task semantics/label quantities from others.

    \item \textit{Trainer clients $\mathcal{T}$} contribute surplus yet fragmentary compute without proprietary data, each controlling a private model segment sized to its capacity.
    As rational, honest-but-curious participants, they do not sabotage the task but may prefer minimal effort for ``free-lunch'' rewards and may be inquisitive about peers’ resources, especially attempting to infer information about $\mathcal{C}$’s data distribution, labels, or sample membership from observed activations.

    \item \textit{Verifier $\mathcal{V}$} is fully trusted and responsible for coordinating and incentive management.
    It can be instantiated by the data client $\mathcal{C}$, a trainer client $\mathcal{T}$, or a neutral third party, endowed with an omniscient audit view that includes the end-to-end activation relay and all segment parameters.
    It has the authority to validate chained watermarks and allocate incentives or penalties, thereby enforcing fair compensation and protocol compliance.
\end{packeditemize}

\subsection{\color{black}Operational Setting and Deployment Context.}
{\color{black}
\textsc{CliCooper} involves a data client $\mathcal{C}$ 
and a set of trainer clients ${\mathcal{T}_i}$.
Instead of relying on centralized servers, the learning pipeline is formed dynamically by chaining multiple trainer segments, each sized according to local capacity. 
A logically separate verifier $\mathcal{V}$ coordinates task instantiation and incentive settlement without participating in training.
This setting naturally captures multi-tenant analytics platforms, edge–cloud learning services~\cite{edge–learning}, and compute marketplaces~\cite{ml_marketplace_34} where participants are elastic, independent, and only partially trusted.

\textsc{CliCooper} treats collaborative model training as a service composed of data provisioning, segmented computation, and verifiable contribution. 
The absence of a centralized training server shifts security and reliability concerns from protocol correctness to service-level risks, such as inference leakage~\cite{unsplit_27}, free-riding~\cite{free-rider}, and ownership disputes~\cite{ownership-dispute} among rational participants. 
Accordingly, the system is designed so that no single trainer can infer task semantics or data properties beyond what is necessary for its assigned segment, while still enabling verifiable attribution of computational contributions. 

These characteristics make \textsc{CliCooper} amenable to extensions, e.g., hybrid SL–FL deployments~\cite{SL-FL}, in which multiple SL pipelines are independent local training services and coordinated at a higher level using federated-style aggregation. 
Supporting this extension does not alter the core trust or privacy assumptions of \textsc{CliCooper}, and only involves modest protocol adaptations, e.g., elevating the verifier to also assume lightweight aggregation responsibilities and extending the chained watermark derivation across aggregation rounds
to preserve the contribution provenance.
In this paper, we focus on the core serverless SL setting and leave a full empirical exploration of hybrid SL–FL deployments to future work.
}

\subsection{Threat Models.}

We distinguish internal users and external adversaries interacting only through exposed APIs. 
All clients are \textbf{rational and honest-but-curious}~\cite{twoPartySL_5}: they share the common interest of completing the main task and thus do not sabotage training, yet they may try to minimize their own effort and may be curious about others’ resources.

\begin{packeditemize}

\item \textbf{Internal users}.
Internal users primarily refer to the trainer clients ${\mathcal{T}_i}$, as the data client $\mathcal{C}$ has minimal visibility without access to others’ parameters, and the verifier $\mathcal{V}$ is fully trusted.
Each $\mathcal{T}_i$ has white-box access to its own segment but sees only \emph{upstream activations} from peers. 
Model parameters of other segments and $\mathcal{C}$’s raw data and true labels remain hidden. 
The dominant risks are curiosity-driven inference and credit/ownership disputes rather than protocol sabotage. 
Typical inference attempts include:

\begin{itemize}
    \item \textit{Clustering attack.} Even without access to $\mathcal{C}$'s raw dataset, $\mathcal{T}$ may apply clustering~\cite{clustering_b, birch, dbscan} to intermediate activation to infer samples of the same class.
    \item \textit{Inversion attack.} Assuming $\mathcal{T}$ has knowledge of $\mathcal{C}$'s model architecture, it may leverage exposed intermediate information to reconstruct the training samples or generate visually similar approximations~\cite{unsplit_27, inversion_26}.
\end{itemize}

{\color{black} We consider a marketplace-style collaboration, where trainer clients are rational and honest-but-curious: they share common interest in completing the main task and thus do not sabotage training but may attempt passive inference from observed activations or seek ``free-lunch'' rewards with minimal effort~\cite{performance-pow}. Unless otherwise stated, trainer clients are assumed non-colluding; even under collusion, trainers still only observe DP-protected activations and operate in the pseudo-label space, while the data client’s raw inputs, true labels, and the secret true$\leftrightarrow$pseudo mapping remain undisclosed. In particular, reusing a pre-trained/open-source segment to bypass training is infeasible because the task semantics are hidden by secret-mapping label expansion and DP perturbation, preventing trainers from identifying a matching pre-trained model in advance.}

{\color{black} We assume that the verifier $\mathcal{V}$ is fully trusted and responsible for settlement: it holds an audit view (end-to-end activation relay and all segment parameters) to verify model utility and chained watermarks and then allocate incentives or penalties. Model runs failing the accuracy requirement are rejected and deemed ineligible for compensation. Notably, $\mathcal{V}$ needs no access to the data client’s raw inputs, true labels, or the secret label mapping; revealing the mapping to any untrusted third party 
should be avoided.}

\item \textbf{External adversaries}.
In contrast, we assume that external adversaries only have black-box~\cite{backdoor_wm_33} access to the released SL model. 
Their sole capability is interacting with the model through its exposed APIs. 
{\color{black}We do not consider white-box access, as it would typically imply server compromise or insider leakage beyond our deployment assumptions.}
Such adversaries may perform model extraction attacks~\cite{model_extract_29}. By querying the API with unlabeled inputs similar to the training data, they can obtain the corresponding labels. The attacker exploits the target SL model as a labeling oracle and subsequently uses the labeled data to train a surrogate model with functionality similar to the original model.

\end{packeditemize}

%% file: sections_1/4_design.tex
\section{Design of \textsc{CliCooper}}\label{sec:design}


\textsc{CliCooper} has five components (Fig.~\ref{framwork}): an SL protocol, secret-mapping label expansion, DP-guarded activations, cooperative training with embedded watermark, and verification.

\begin{figure*}[!t]
        \centering
        \includegraphics[width=0.8\linewidth]{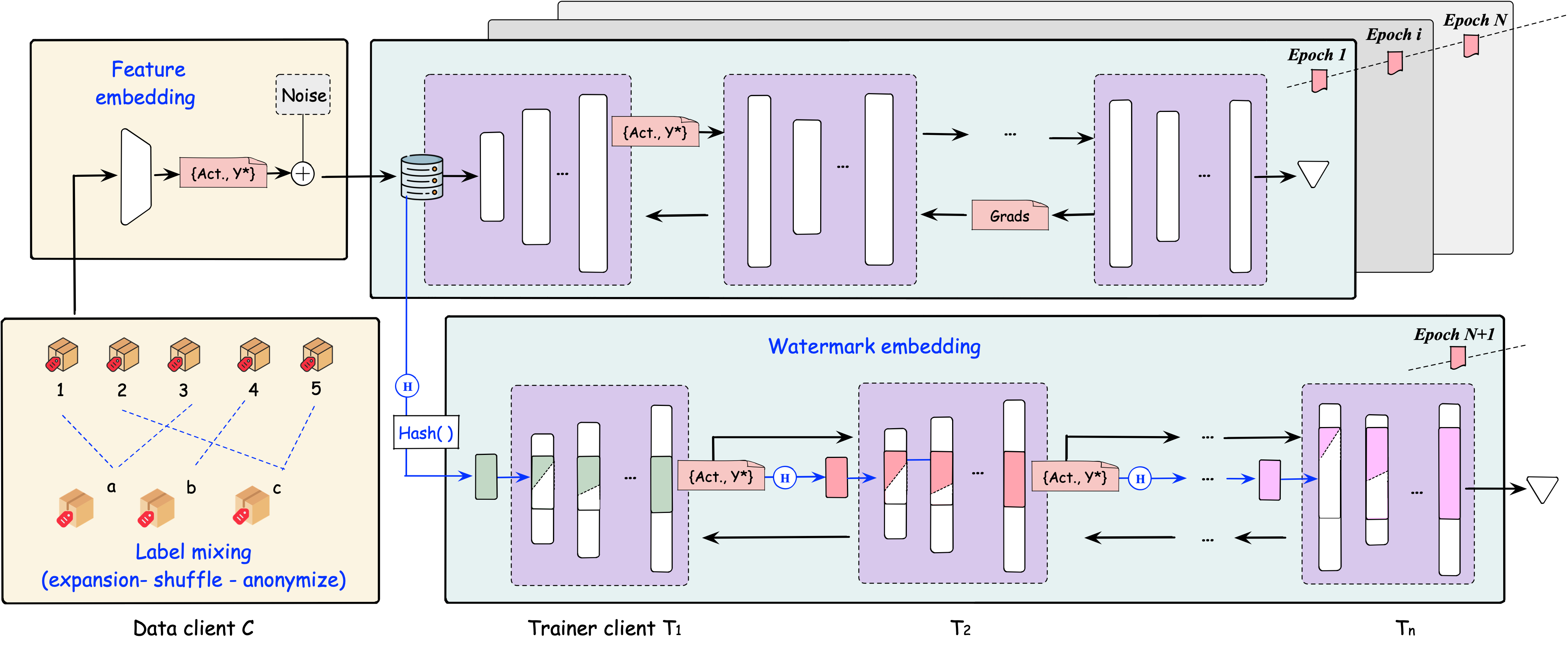}
        \caption{\textbf{\textsc{CliCooper} design:}
        The data client $\mathcal{C}$ holds the dataset but has very limited compute and thus collaborates with multiple trainer clients $\mathcal{T}$ for split learning. First, to disclose the ground-truth labels $Y$, $\mathcal{C}$ expands $Y$ to pseudo-labels $Y^*$ whose quantity and semantics differ from $Y$. Next, $\mathcal{C}$ uses a lightweight encoder to embed features and applies DP to protect activations. $\mathcal{T}$ then train the model using DP-protected activations $\mathbb{M}_{\mathcal{C}}^{\text{\tiny DP}}(W_{\mathcal{C}},\mathbb{D}^*)$ and $Y^*$ for $N$ epochs until convergence. In the $N{+}1$th epoch, each trainer $\mathcal{T}$ embeds a chained watermark into its assigned subnetwork in a pipelined manner, asserting ownership claim for their contributions to access compensation.
        }
        \label{framwork}
\end{figure*}

\subsection{SL Protocol}
\label{subsec: sl_protocol}

We target a marketplace-style collaboration where heterogeneous clients provide training services in exchange for rewards. 
Clients with higher capacity host larger network partitions, while fragmentary devices host smaller segments. 
Unlike conventional SL that delegates downstream training to a single powerful server, \textsc{CliCooper} distributes the “server side’’ over multiple trainer clients, scaling the number and size of segments with model complexity and available compute.

\noindent\textbf{Pre-training negotiation.} Before any data flow, the task publisher (usually $\mathcal{C}$ or $\mathcal{V}$) conducts a light-weight negotiation phase with participating trainers to (i) agree on split points and a strict execution order $\mathcal{T}_1 \!\rightarrow\! \dots \!\rightarrow\! \mathcal{T}_n$, and (ii) size each segment to match the host’s compute/memory budget and the end-to-end latency target. 
{\color{black}Each trainer stores and updates only its own segment parameters locally, and does not share them with peers. }
The agreed plan fixes interface tensor shapes and optimizer/hyper-parameter responsibilities for each segment.

\noindent\textbf{Training flow.} Let $\mathbb{D}$ denote the task dataset held exclusively by $\mathcal{C}$, and $W=\{W_{\mathcal{C}},W_{\mathcal{T}_1},\dots,W_{\mathcal{T}_n} \}$ denote the parameters contributed by $\mathcal{C}$ and $\mathcal{T}_1,\dots, \mathcal{T}_n$. The end-to-end system workflow respects the relay structure as follows~\cite{unsplit_27}:
\begin{equation*}
    \small
    \mathbb{M}(W,\mathbb{D})\!=\!\mathbb{M}_{\mathcal{T}_n}\!\big(W_{\mathcal{T}_n},\,\mathbb{M}_{\mathcal{T}_{n-1}}\!(\dots\! \mathbb{M}_{\mathcal{T}_1}(W_{\mathcal{T}_1},\,\mathbb{M}_{\mathcal{C}}^{\text{\tiny DP}}(W_{\mathcal{C}},\mathbb{D}^*)))\big).
\end{equation*}
Here, only activations $\mathbb{M}(\cdot)$ are relayed downstream, {\color{black} according to bandwidth/latency-related coordination}, while model parameters remain local and opaque to peers. Backpropagation proceeds locally at each trainer over its segment using the received upstream gradients. 
No true labels or raw inputs are disclosed to trainers, as $\mathbb{M}_{\mathcal{C}}^{\text{\tiny DP}}(\cdot)$ is DP-protected activations and $\mathbb{D}^*$ indicates $\mathbb{D}$'variant after an expansion process (\S\ref{subsec: expansion_Gaussian}).

To enable market settlement and dispute resolution, trainers agree to embed a chained watermark lineage during training (\S\ref{subsec: Chained Watermarking}); segments lacking verifiable marks are ineligible.

\subsection{Secret-mapping Expansion and DP-based Protection}
\label{subsec: expansion_Gaussian} 


\textsc{CliCooper} answers \textcolor{violet}{\textbf{RQ1} (data privacy)} and \textcolor{violet}{\textbf{RQ3} (unauthorized-use defense)} by coupling secret-mapping label expansion with DP-guarded activations. The data client $\mathcal{C}$ never discloses true labels; instead, it obfuscates task semantics and label cardinality while protecting raw inputs against inversion attack (cf.~\textbf{Algorithm}~\ref{algo:label_expansion}).

\begin{algorithm}[t]
    \linespread{1.1}
    \footnotesize
    \caption{Label Expansion \& Privacy Protection}
    \begin{algorithmic}[1]
    \Require 
    $\mathbb{D}$ (\textcolor{violet}{original dataset}), 
    $Y$ (\textcolor{violet}{true label}),
    $\mathcal{G}_Y$ (\textcolor{violet}{predefined one-to-many map})),
    $\varepsilon$ (\textcolor{violet}{privacy budget of DP noise}), $\{ g_i\}_{i=1}^{q}$ (\textcolor{violet}{expansion factor for each true label}). 

    \Ensure $\mathbb{M}_{\mathcal{C}}^{\text{\tiny DP}}(W_{\mathcal{C}},\mathbb{D}^*)$ (\textcolor{violet}{DP-protected activations}),
    $Y^*$ (\textcolor{violet}{pseudo-labels}).
    
    \For{$Y_i \in Y$} \Comment{label expansion loop}
        \State $Y_i \mapsto \{Y_{i_1}^*,\dots,Y_{i_{g_i}}^*\}$,
        \Comment{true labels expand to pseudo-labels}
        
        \State $\{(\mathbb{D}_{i},Y_i)\}\ \mapsto\ \{(\mathbb{D}_{i_1}^*,Y^*_{i_1})\!\dots (\mathbb{D}_{i_{g_i}}^*,Y^*_{i_{g_i}})\}$
        \Comment{data-label pairing.}

    \EndFor 

    \State $Y \xrightarrow{\ \mathcal{G}_Y\ } Y^*,\ Y^*\!=\!\{Y_1^*,\dots,Y_{\sum_{i=1}^{q} g_i}^*\}.$ \Comment{shuffle randomly.}

    \State $\{\mathbb{D},Y\}\ \mapsto\ \{\mathbb{D}^*,Y^*\}.$ \Comment{transform accordingly.}

    \State $\tilde{A}=\mathbb{M}_{\mathcal{C}}(W_{\mathcal{C}},\mathbb{D}^*)$.
    \Comment{generate client-side activations.}

    \State $\tilde{A}\ \!\xrightarrow{\ \textsf{Clip}_{S}\ }\ \! \bar{A}$.
    \Comment{enforce an $\ell_{1}$ sensitivity bound.}

    \State $\bar{A}\ \!\xrightarrow{\ \mathbb{G}(\varepsilon)\ }\ \!
\mathbb{M}_{\mathcal{C}}^{\text{\tiny DP}}(W_{\mathcal{C}},\mathbb{D}^*) 
\!=\! \bar{A} + \mathrm{Laplace}\!\left(0,\; \tfrac{\Delta_{1}}{\varepsilon}\right)$.
    \Comment{Laplace-based DP mechanism injection.}
    
    \State \textbf{RETURN} $\{\mathbb{M}_{\mathcal{C}}^{\text{\tiny DP}}(W_{\mathcal{C}},\mathbb{D}^*)$, $Y^*\}$
    \end{algorithmic}
    \label{algo:label_expansion}
    \end{algorithm}

\subsubsection{\underline{Secret-mapping label expansion}}
\label{subsec: label_expansion} 
To hide both label semantics and quantity while preserving trainability, we employ a \textbf{secret-mapping} expansion controlled by the data client $\mathcal{C}$.
Let $Y=\{Y_1,\dots,Y_q\}$ be true labels. Data client $\mathcal{C}$ samples a secret one-to-many map $\mathcal{G}_Y$ and expands each true class $Y_i$ to $g_i$ $(g_i \geq 1)$ pseudo-classes, $Y_i \mapsto \{Y_{i_1}^*,\dots,Y_{i_{g_i}}^*\}$, yielding $Y^*$ with $\sum_{i=1}^{q} g_i$ categories (Lines 1--3):
\begin{equation}
\small
Y \xrightarrow{\ \mathcal{G}_Y\ } Y^*,\ Y^*\!=\!\{Y_1^*,\dots,Y_{\sum_{i=1}^{q} g_i}^*\},
\label{label_mapping}
\end{equation}
where $Y^*=\bigcup_{i=1}^{q}\{Y_{i_1}^*,\dots,Y_{i_{g_i}}^*\}$ is the union of all per-class pseudo-labels. The ordering is randomly shuffled and recorded.
Data-and-label pairs are remapped accordingly (Line 4): 
\begin{equation}
\small
\{(\mathbb{D}_{i},Y_i)\}\ \mapsto\ \{(\mathbb{D}_{i_1}^*,Y^*_{i_1})\!\dots (\mathbb{D}_{i_{g_i}}^*,Y^*_{i_{g_i}})\},
\label{data_mapping}
\end{equation}
where $\{\mathbb{D}_{i_1}^*,\dots,\mathbb{D}_{i_{g_i}}^*\}$ are obtained by partitioning and noise-aware augmentation of $\mathbb{D}_i$~\cite{data_augment, data_augment_text}; class priors and balance remain stable under expansion.
The dataset–label pair becomes
\begin{equation}
\small
\{ (\mathbb{D},Y)\}\ \mapsto\ \{ (\mathbb{D}^*,Y^*)\},
\label{data_mapping}
\end{equation}
where $\mathbb{D}^*$ and $Y^*$ are the expanded, anonymized, and permuted versions of the original dataset and labels, with $\mathbb{D}^*$ transformed accordingly by  $\mathcal{G}_Y$.
Define $\gamma = \frac{\sum_{i=1}^{q} g_i}{q} > 1$.

This expansion does not harm utility for authorized users because training is performed consistently on the expanded label space, so the loss is invariant to a fixed mapping~\cite{yu2023split}.
Authorized parties who legally possess $\mathcal{G}_Y$ perform \textbf{demasking} at the inference with the inverse mapping $\mathcal{G}_Y^{-1}$:
\begin{equation}
\small
\{Y_1^*,\dots,Y_{\sum g_i}^*\}\xrightarrow{\ \mathcal{G}_Y^{-1}\ }\{Y_1,\dots,Y_q\}.
\label{label_demapping}
\end{equation}
This binds model utility to authorization and obfuscates both semantics and quantities.
Without $\mathcal{G}_Y$ or $\mathcal{G}_Y^{-1}$, outputs become operationally degraded for unauthorized users, impeding functional repurposing and correct re-labeling of downstream data.

\subsubsection{\underline{DP-protected activations}}
\label{subsec: Gaussian_protection}
To curb activation inversion and sensitive-attribute inference, $\mathcal{C}$ perturbs smashed activations with calibrated DP noise before upload. 
Let $\tilde{A}=\mathbb{M}_{\mathcal{C}}(W_{\mathcal{C}},\mathbb{D}^*)$ be the client-side activations after the forward pass with the expanded dataset $\mathbb{D}^*$. 
We enforce a $\ell_{1}$ sensitivity bound via \emph{per-layer $\ell_{1}$-clipping} to radius $S$, yielding $\Delta_{1}\!\le\!2S$ for neighboring inputs.
We add zero-mean Laplace\footnote{\color{black}Our framework is mechanism-agnostic, Gaussian ($\varepsilon$,$\delta $)-DP can be 
replaced to $\ell_{1}$ clipping + Gaussian calibration, if one prefers approximate DP.} noise to $\Delta_{1}$ and the target privacy budget $\varepsilon$ (Lines 5--7):
\begin{equation*}
\small
\tilde{A}\ \!\xrightarrow{\ \textsf{Clip}_{S}\ }\ \! \bar{A}
\text{, and }
\bar{A}\ \!\xrightarrow{\ \mathbb{G}(\varepsilon)\ }\ \!
\mathbb{M}_{\mathcal{C}}^{\text{\tiny DP}}(W_{\mathcal{C}},\mathbb{D}^*) 
\!=\! \bar{A} + \mathrm{Laplace}\!\left(0,\; \tfrac{\Delta_{1}}{\varepsilon}\right),
\end{equation*}
where $\mathbb{G}(\varepsilon)$ denotes a randomized Laplace mechanism~\cite{phan2017adaptive} with privacy budget $\varepsilon$. 

\begin{thm}[DP-protected activations]
\label{thm:dp_protected_activations}
Let $\bar A(x)=\textsf{Clip}^{(1)}_{S}\!\big(\mathbb{M}_{\mathcal{C}}(W_{\mathcal{C}},x)\big)\in\mathbb{R}^d$ be the per-sample activation after per-layer $\ell_{1}$-clipping with radius $S$. The $\ell_{1}$-sensitivity satisfies
$\Delta_{1}=\sup_{x\sim x'}\|\bar A(x)-\bar A(x')\|_{1}\le 2S$.
Consider the Laplace mechanism $
\mathcal{M}_{\varepsilon}(x)=\bar A(x)+Z, Z\stackrel{\text{i.i.d.}}{\sim}\mathrm{Laplace}\!\left(0,\tfrac{\Delta_{1}}{\varepsilon}\right),$
yielding pure $\varepsilon$-DP. For any fixed DP activation $y_a=\mathcal{M}_{\varepsilon}(x_a)$, any two DP activations $y_b,y_c$, and candidate class $t$,
\begin{equation}
\frac{\Pr\big[t\ \text{gen} (y_a \wedge y_b)\big]}
     {\Pr\big[t\ \text{gen} (y_a \wedge y_c)\big]}
\;\le\; e^{2\varepsilon},
\label{eq:pair_ratio_dp}
\end{equation}
where $\,t\ \text{generated }(y_a \wedge y_b)\,$ denotes the event that both DP activations $y_a$ and $y_b$ were produced from inputs whose (hidden) true class is $t$.
\end{thm}

\begin{prf} The per-coordinate Laplace noise with scale $\Delta_{1}/\varepsilon$ ensures that for any input $x$ and any observation $y$, the likelihood ratio is bounded
$e^{-\varepsilon}\le \frac{\Pr[y\,|\,x]}{\Pr[y\,|\,x']}\le e^{\varepsilon}$,
$\forall x\!\sim \!x'$, according to the principle of the standard Laplace $\varepsilon$-DP~\cite{phan2017adaptive}.
Fix a candidate class $t$ and write the joint ``same-class'' event under a product model:
$\Pr[t\text{ gen }(y_a\wedge y_b)]=\sum_{x\in t}\Pr[y_a|x]\Pr[y_b|x]\Pr[x]$,
and similarly for $(y_a\wedge y_c)$:
$\Pr[t\text{ gen }(y_a\wedge y_c)]=\sum_{x\in t}\Pr[y_a|x]\Pr[y_c|x]\Pr[x]$.
Then, 
\begin{align}
\frac{\Pr[t\text{ gen }(y_a\wedge y_b)]}{\Pr[t\text{ gen }(y_a\wedge y_c)]}
&=
\frac{\sum_{x\in t}\Pr[y_a|x]\Pr[y_b|x]\Pr[x]}
     {\sum_{x\in t}\Pr[y_a|x]\Pr[y_c|x]\Pr[x]}
\nonumber\\
&\le
\frac{\sum_{x\in t}\Pr[y_a|x]\,(e^{\varepsilon}\Pr[y_a|x])\Pr[x]}
     {\sum_{x\in t}\Pr[y_a|x]\,(e^{-\varepsilon}\Pr[y_a|x])\Pr[x]}
\nonumber\\
&= e^{2\varepsilon},
\end{align}
where the $\varepsilon$-DP
bound 
on $\Pr[y_b|x]/\Pr[y_a|x]$ in the numerator and $\Pr[y_c|x]/\Pr[y_a|x]$ in the denominator yields  $e^{\varepsilon}$ and $e^{-\varepsilon}$, respectively.
Hence, \eqref{eq:pair_ratio_dp} holds.
\qedhere
\end{prf}

As demonstrated in Theorem~\ref{thm:dp_protected_activations}, under $\ell_{1}$-clipping with Laplace noise, no observer can improve the odds that two DP activations $(y_a,y_b)$ came from the same candidate class $t$ relative to $(y_a,y_c)$ by more than $e^{2\varepsilon}$. This caps pairwise linkage attacks and underpins the secrecy of class membership.

Given $\mathcal{C}$’s limited compute and to avoid cumulative privacy loss from repeated composition, both label-space expansion and DP perturbation are performed once before collaborative training. That is, $\mathcal{C}$ expands labels, computes a single forward pass over $\mathbb{D}^*$ to produce $\tilde{A}$, applies clipping and the Laplace mechanism to obtain $\mathbb{M}_{\mathcal{C}}^{\text{\tiny DP}}$, and caches these DP-guarded activations for the further training rounds.
In the subsequent pipeline, $\mathcal{C}$ \emph{sends only} $\mathbb{M}_{\mathcal{C}}^{\text{\tiny DP}}$ to the first trainer $\mathcal{T}_1$ and the pseudo-labels $Y^*$ to the last trainer $\mathcal{T}_n$, keeping raw inputs and true labels remain undisclosed.

\subsection{Training with Chained Watermarking}
\label{subsec: Chained Watermarking}

Since trainers $\mathcal{T}$ are only partially trusted, they must demonstrate that they actually performed training and that they own the layers they optimized—rather than submitting a ``free-lunch'' or open-source pretrained model to claim rewards. \textsc{CliCooper} answers \textcolor{violet}{\textbf{RQ2} (layer ownership)} by enforcing pipeline training with chained watermarks embedded into the $\mathcal{T}$-side submodels (cf. \textbf{Algorithm}~\ref{algo:watermark_embedding}). 
Once a trainer $\mathcal{T}_i$ brings its submodel to convergence, it must embed a watermark $\Lambda_{\mathcal{T}_i}$ to attest ownership; crucially, $\Lambda_{\mathcal{T}_i}$ is not chosen by $\mathcal{T}_i$ but is deterministically derived from the predecessor’s output activation via a collision-resistant hash. By seeding each watermark with 
$\Lambda_{\mathcal{T}_i}{\leftarrow} \mathbb{H} (\cdot)$ ($\mathbb{H} (\cdot)$ is a hash function to generate a watermark), the scheme prevents precomputation or substitution attacks; a trainer cannot cheat without executing the prescribed protocol.

\subsubsection{\underline{Pipeline training}}
Before the cooperative training, the data client $\mathcal{C}$ processes the raw dataset with a pre-trained split encoder and safeguards the resulting activations via a DP mechanism, producing DP-protected activations $\mathbb{M}_{\mathcal{C}}^{\text{\tiny DP}}$. Once $\mathcal{C}$ uploads the activation 
$\mathbb{M}_{\mathcal{C}}^{\text{\tiny DP}}$ to $\mathcal{T}_1$ and the pseudo-labels $Y^*$ to $\mathcal{T}_n$, it ceases participation in the subsequent pipeline training. The trainers $\mathcal{T} = \{\mathcal{T}_1, \dots,\mathcal{T}_n\}$ then execute a predefined SL protocol, proceeding sequentially in a pipeline (Lines 1--3).

During each training epoch, $\mathcal{T}_1$ consistently consumes the $\mathcal{C}$-provided activation 
$\mathbb{M}_{\mathcal{C}}^{\text{\tiny DP}}$ as input and forwards its output to $\mathcal{T}_2$. Each intermediate trainer $\mathcal{T}_i$ for $i=2,3..n-1$ takes as input the activation produced by $\mathcal{T}_{i-1}$, performs its designated computation, and passes the resulting activation to $\mathcal{T}_{i+1}$. Upon receiving the activation from $\mathcal{T}_{n-1}$, 
$\mathcal{T}_n$ combines it with the pseudo-labels 
$Y^*$ supplied by the client $\mathcal{C}$ to compute the loss, perform gradient descent, and backpropagate gradients upstream through 
$\mathcal{T}_{n-1}, \mathcal{T}_{n-2}...\mathcal{T}_{1}$. The trainers repeat this process for $N$ rounds until the SL model converges.

{\color{black}Watermark embedding is performed after the pipeline training converges: after $N$ epochs of standard SL optimization, each trainer embeds its watermark in the ($N+1$)-th epoch via short regularized watermarking that is monitored until the watermark detection rate reaches the threshold. Therefore, watermarking does not affect the gradient flow or convergence of the main task training in the first $N$ epochs. Model accuracy with watermarking remains close to the baseline across different watermark lengths. In addition, embedding/verification overhead is negligible compared to training time.}


\subsubsection{\underline{Watermark generation}}
\label{subsec: Watermark generation}

\begin{algorithm}[t]
    \linespread{1.1}
    \footnotesize
    \caption{Training \& Chained  Watermarking}
    \begin{algorithmic}[1]
    \Require $\mathbb{M}_{\mathcal{C}}^{\text{\tiny DP}}(W_{\mathcal{C}},\mathbb{D}^*)$ (\textcolor{violet}{DP-protected activation}), 
    $W_{\mathcal{C}}$ (\textcolor{violet}{the pretrained model of $\mathcal{C}$}),
    $ID_{\mathcal {T}_i}$ (\textcolor{violet}{$\mathcal {T}_i$ node identity}), 
    $\mu_i$ (\textcolor{violet}{secret nonce}),
    $\eta_G$ (\textcolor{violet}{threshold for successful watermark embedding}), 
    $l_w$ (\textcolor{violet}{main task loss function}),
    $l_\Lambda$ (\textcolor{violet}{watermark embedding regularizers}).
    
    \Ensure $W=\{W_{\mathcal{C}},W_{\mathcal{T}_1},\dots,W_{\mathcal{T}_n} \}$ (\textcolor{violet}{SL model with watermark}).

    \While{ ($\lnot$ converging)} \Comment{until convergence}

        \State $W_{\mathcal{C}}\overset{\mathbb{M}_{\mathcal{C}}^{\text{\tiny DP}}}{\rightarrow}  W_{\mathcal{T}_1}\overset{\mathbb{M}_{\mathcal{T}_1}}{\rightarrow}\dots \overset{\mathbb{M}_{\mathcal{T}_{n-1}}}{\rightarrow}W_{\mathcal{T}_n} $ 
        \Comment{pipeline training.}
        
    \EndWhile
    
    \State $\mathbb{M}_{\mathcal{C}}^{\text{\tiny DP}}$, $\mu_i$, $i$ \Comment{initial activation, and trainer index.}

     \For{$i=1,2\dots n$} \Comment{trainer's watermark embedding loop}

        \State $Z_{\mathcal {T}_i} = \textsf{WMPosition}(\mu_i)$
        \Comment{select matrix for weights used for $\Lambda_{\mathcal {T}_i}$}

        \State $\Lambda_{\mathcal{T}_i}, k_{\mathcal{T}_i} = \textsf{WMGen}(\mathcal{H}_i), \textsf{KeyGen}(\mathcal{H}_i)$ where 
            \If{$i=1$}

            \State  $\mathcal{H}_1=\mathbb{H}(\mathbb{M}_{\mathcal{C}}^{\text{\tiny DP}}, 1, \mu_1, ID_{\mathcal{T}_1})$
            \Else
           \State $\mathcal{H}_i = \mathbb{H}(\!\mathbb{M}_{\mathcal{T}_{i-1}},\;i ,\; \mu_i,\; ID_{\mathcal{T}_{i}})$
        \EndIf \Comment{calculate trainer's watermark}
     
        \While{$\eta < \eta_G$} \Comment{until watermark convergence}

        \State $W_{\mathcal{T}_i} = \arg\min_{W} \left( l_w(W) + \lambda l_\Lambda(W) \right)$ \Comment{$\Lambda_{\mathcal {T}_i}$ embedding}

            \State $\eta = 1-\frac{ \sum_{b}^{B} (\mathbb{O}(W_{\mathcal{T}_i},Z_{\mathcal{T}_i},k_{\mathcal{T}_i}
    )[j]\ne \Lambda_{\mathcal{T}_i}[j])}{B},$ \Comment{update $\eta$}

        \If{$\eta \geq \eta_G$}

            \State  Save $W_{\mathcal{T}_i}$
            
        \EndIf
        
        \EndWhile

        \State $i \ += 1$
    \EndFor

    \State \textbf{RETURN} $W=\{W_{\mathcal{C}},W_{\mathcal{T}_1},\dots,W_{\mathcal{T}_n} \}$
    \end{algorithmic}
    \label{algo:watermark_embedding}
    \end{algorithm}
  
Watermark embedding follows pipeline training and proceeds first with the generation of the watermark.
Watermark generation begins with the task publisher (which may be either $\mathcal {C}$ or the $\mathcal {V}$) assigning each $\mathcal {T}_i$ node a fixed nonce $\mu_i$, which $\mathcal {T}_i$ use as the basis for generating the watermark $\Lambda_{\mathcal{T}_i}$, the embedding key $k_{\mathcal{T}_i}$, and the position matrix $Z_{\mathcal{T}_i}$ that specifies where $\Lambda_{\mathcal{T}_i}$ is embedded.
{\color{black}Concretely, trainer $\mathcal{T}_i$ embeds the watermark only into the parameters of its own assigned layers. We first flatten the trainable tensors in $\mathcal{T}_i$’s segment into a 1-D vector and then use $Z_{\mathcal{T}_i}$ as a selection mask to pick the corresponding weight entries for watermark embedding.}
$Z_{\mathcal{T}_i}$ is generated solely from the fixed seed $\mu_i$, whereas 
$\Lambda_{\mathcal{T}_i}$ and $k_{\mathcal{T}_i}$ are derived from the predecessor’s output jointly with 
$k_{\mathcal{T}_i}$ and the trainer identity $ID_{\mathcal{T}_i}$ (Lines 4--12):
\begin{equation}
\begin{aligned}
Z_{\mathcal{T}_i} &= \textsf{WMPosition}(\mu_i), i=1,2...n \\
\Lambda_{\mathcal{T}_i},\; k_{\mathcal{T}_i} &= \textsf{WMGen}(\mathcal{H}_i),\; \textsf{KeyGen}(\mathcal{H}_i) \text{ with} \\
\mathcal{H}_i &= \mathbb{H}(\!\mathbb{M}_{\mathcal{T}_{i-1}},\;i ,\; \mu_i,\; ID_{\mathcal{T}_{i}}), i=2,3..n, \\
 \text{while} \
\mathcal{H}_1 &= \mathbb{H}(\mathbb{M}_{\mathcal{C}}^{\text{\tiny DP}},\;1 ,\; \mu_1,\; ID_{\mathcal{T}_1}), i=1,
\end{aligned}
\label{wm_k_gen}
\end{equation}
where $\textsf{WMPosition}(\cdot)$, $\textsf{WMGen}(\cdot)$, and $\textsf{KeyGen}(\cdot)$ are deterministic functions, with all randomness derived from the publisher-fixed nonce $\mu_i$ and $\mathcal {T}_i$ identity $ID_{\mathcal{T}_i}$. This ensures only the legitimate $\mathcal {T}_i$ can generate valid $\Lambda_{\mathcal{T}_i}$ and $k_{\mathcal{T}_i}$, while the $\mathcal {T}$'s index $i$ enforces sequential linkage.
This chained derivation of the watermark and embedding keys thwarts precomputation: a trainer cannot anticipate ($\Lambda_{\mathcal{T}_i}$, $k_{\mathcal{T}_i}$) before observing the predecessor’s activation, nor embed a valid watermark into an unauthorized model in advance.

\subsubsection{\underline{Watermark embedding}} 
\label{subsec: Watermark embedding}
Since $\mathcal {T}$ possesses no local data, it is unable to embed backdoor-based watermarks~\cite{backdoor_wm_33}. Instead, watermarking approaches that are integrated into the training process by modifying model features (e.g., EWDNN~\cite{uchida2017embedding_18}, RIGA~\cite{wang2021riga_19}, FedIPR~\cite{li2022fedipr_20}) represent the most suitable schemes for watermark embedding in this setting.

After $N$ epoch training, the model converges. In the $(N+1)$-th epoch, the unique watermark $\Lambda_{\mathcal{T}_i}$ is embedded into the weights $W_{\mathcal{T}_i}$ according to the matrix $Z_{\mathcal{T}_i}$, gradually forming the weights of the final epoch, where $\Lambda_i$ has been successfully embedded.
This embedding process in multiple rounds in the $(N+1)$th epoch can be written as (Lines 13--14):
\begin{equation}
W_{\mathcal{T}_i} = \arg\min_{W} \left( l_w(W) + \lambda l_\Lambda(W) \right) 
,
\end{equation}
where $l_w$ is the main task loss function, $l_\Lambda$ is the watermark embedding regularizer, and $\lambda$ is a hyper-parameter that balances the main task training and watermark embedding.

While embedding the watermark $\Lambda_{\mathcal{T}_i}$, the process is actively monitored in each round of the $(N+1)$th epoch. Simultaneously, the watermark is extracted from $W_{\mathcal{T}_i}$ and compared to the expected watermark. The watermark detection rate $\eta$ is calculated using the Hamming distance (Line 15): 
\begin{equation}
    \eta = 1-\frac{ \sum_{b}^{B} (\mathbb{O}(W_{\mathcal{T}_i},Z_{\mathcal{T}_i},k_{\mathcal{T}_i}
    )[j]\ne \Lambda_{\mathcal{T}_i}[j])}{B}, 
\end{equation}
where $\mathbb{O}(\cdot)$ is the watermark extraction function, and $B$ is the size of $\Lambda_{\mathcal{T}_i}$. Once $\eta$ reaches or exceeds the predefined threshold $\eta_G$ (i.e., $\eta\geq \eta_G$), the embedding process for the current trainer $\mathcal{T}_i$ is considered complete. 
Then $\mathcal{T}_i$ freezes and saves its checkpoint (Lines 16--17) and outputs activation $\!\mathbb{M}_{\mathcal{T}_i}$ to $\mathcal{T}_{i+1}$, who will follow the same watermark generation and embedding process as for the $\mathcal{T}_i$-th trainer.


{\color{black}The chained watermark is sequential across trainers by design: each trainer’s watermark/key is deterministically derived from the predecessor’s output activation, together with a publisher-fixed nonce and identity, and the trainer index enforces sequential linkage.} This iterative process continues until all the trainers finish watermark embedding (Lines 5--21).
It ensures that each watermark is uniquely bound to its predecessor trainer's output, thereby preventing $\mathcal {T}$ from fabricating a valid contribution beforehand to access compensation.

\subsection{Chained Watermark Verification}
\label{subsec: Watermark verification}
Prior to verification, client $\mathcal{C}$ uploads the pretrained client-side model $W_{\mathcal{C}}$ together with $\mathbb{M}_{\mathcal{C}}^{\text{\tiny DP}}$ to the verifier, while each trainer $\mathcal{T}_i$ uploads its checkpoints $W_{\mathcal{T}_i}$. The verifier $\mathcal{T}$ assembles the complete SL model $W$ by merging $W_{\mathcal{T}_i}$ and, using $W_{\mathcal{T}_i}$ to calculate $\mathbb{M}_{\mathcal{T}_i}$ for subsequent watermark generation. The verification procedure consists of two stages: model performance verification and chained watermark verification.

\begin{algorithm}[t]
    \linespread{1.1}
    \footnotesize
    \caption{Watermark Verification}
    \begin{algorithmic}[1]
    \Require $W=\{W_{\mathcal{C}},W_{\mathcal{T}_1},\dots,W_{\mathcal{T}_n} \}$ (\textcolor{violet}{SL model with watermark}), 
    $\mathbb{M}_{\mathcal{C}}^{\text{\tiny DP}}(W_{\mathcal{C}},\mathbb{D}^*)$ (\textcolor{violet}{DP-protected activation}),
    $ID_{\mathcal{T}_i}$ (\textcolor{violet}{$\mathcal{T}_i$ node identity}), 
    $\mu_i$ (\textcolor{violet}{secret nonce}), 
    $\mathbb{D}_t $ (\textcolor{violet}{test dataset}), $\eta_{G}$(\textcolor{violet}{threshold of similarity between the hashed and extracted watermarks}). 
    \Ensure ``\textbf{Fail}'' or ``\textbf{Success}''.
    

    \State $\mathcal{V}$ receives $W=\{W_{\mathcal{C}},W_{\mathcal{T}_1},\dots,W_{\mathcal{T}_n} \}$, $ID_{\mathcal{T}_i}$ from $\mathcal{T}$ and $\mathcal{C}$.

    \State $\hat{{Acc}}_{{main}}=\textsf{Test}(W, \mathbb{D}_t)$ \Comment{test the model main task accuracy} 

    \State \textbf{if} $\hat{Acc}_{main}$ not meeting expectations, \textbf{then}
        \State \hspace{4mm} \textbf{RETURN} ``\textbf{Fail}'', \textbf{break}

    
    \For{$i=1,2\dots n$} \Comment{verification loop}

        \State $Z_{\mathcal {T}_i} = \textsf{WMPosition}(\mu_i)$. \Comment{recover the matrix for $\Lambda_{\mathcal{T}_i}$}
        
        \State $\Lambda_{\mathcal{T}_i}, k_{\mathcal{T}_i} = \textsf{WMGen}(\mathcal{H}_i), \textsf{KeyGen}(\mathcal{H}_i)$ where 
            \If{$i=1$}

            \State  $\mathcal{H}_1=\mathbb{H}(\mathbb{M}_{\mathcal{C}}^{\text{\tiny DP}}, 1, \mu_1, ID_{\mathcal{T}_1})$
            \Else
           \State $\mathcal{H}_i = \mathbb{H}(\!\mathbb{M}_{\mathcal{T}_{i-1}},\;i ,\; \mu_i,\; ID_{\mathcal{T}_{i}})$
        \EndIf
             \Comment{calculate the watermark from the previous model by $\mathbb{H}(\cdot)$}
            
        \State $\hat{\Lambda}_{\mathcal{T}_i} = \mathbb{O}(W_{\mathcal{T}_i},Z_{\mathcal{T}_i},k_{\mathcal{T}_i})$ \Comment{extract the current watermark}
        
        \State $\eta_{\mathcal{T}_i} = 1-\frac{ \sum_{b}^{B} (\hat{\Lambda}_{\mathcal{T}_i}[j]\ne \Lambda_{\mathcal{T}_i}[j])}{B}$ \Comment{calculate the $\Lambda_{\mathcal{T}_i}$ similarity}
        
        \State \textbf{if} $\eta_{\mathcal{T}_i} < \eta_G$ \textbf{then}
        \State \hspace{4mm} \textbf{RETURN} ``\textbf{Fail}'', \textbf{break}
    \EndFor 
    \State \textbf{RETURN} ``\textbf{Success}''
    \end{algorithmic}
    \label{algo:watermark_verification}
    \end{algorithm}

The process begins with the verifier $\mathcal {V}$ evaluating the checkpoint $W$ (Lines 1--4). This step ensures that the SL model achieves an accuracy level above a predefined threshold, which is determined based on empirical values commonly adopted in ML marketplaces~\cite{ml_marketplace_34}. If the accuracy of $W$ falls below this threshold, $\mathcal {V}$ rejects $\mathcal {T}$’s training as invalid.
Once the SL model $W$ passes the performance evaluation,  $\mathcal {V}$ proceeds to chained watermark verification (Lines 5--17). Using~\eqref{wm_k_gen},  $\mathcal {V}$ employs the activation $\mathbb{M}_{\mathcal{T}_i}$ of the model $W_{\mathcal{T}_i}$ to regenerate the ground-truth watermark $\Lambda_i$, the embedding key $ k_i$, and the position matrix $Z_{\mathcal{T}_i}$ required for verifying the checkpoint $W_{\mathcal{T}_i}$ (Lines 6--12). 
The ground-truth watermark for $\mathcal {T}_1$ is computed by feeding $\mathbb{M}_{\mathcal{C}}^{\text{\tiny DP}}$ into~\eqref{wm_k_gen}.
This step ensures $\Lambda_{\mathcal{T}_i}$ is legitimately derived from $\mathbb{M}_{\mathcal{T}_i}$ of the model $W_{\mathcal{T}_i}$. Then, $\mathcal{V}$ uses $Z_{\mathcal{T}_i}$ and $k_{\mathcal{T}_i}$ on the checkpoint model $W_{\mathcal{T}_i}$ to extract the embedded watermark $\hat{\Lambda}_{\mathcal{T}_i}$ (Line 13):
    $\hat{\Lambda}_{\mathcal{T}_i} = \mathbb{O}(W_{\mathcal{T}_{\mathcal{T}_i}},Z_{\mathcal{T}_i},k_{\mathcal{T}_i})$. 

Given a predefined watermark detection threshold $\eta_G$, the verification for the checkpoint $W_{\mathcal{T}_i}$ is considered successful if the similarity between the extracted watermark $\hat{\Lambda}_i$ and the ground-truth watermark $\Lambda_i$ exceeds $\eta_G$. The watermark detection rate of $\hat{\Lambda}_i$ can be formally expressed as (Line 14):
\begin{equation}
    \eta_{\mathcal{T}_i} = 1-\frac{ \sum_{b}^{B} (\hat{\Lambda}_{\mathcal{T}_i}[j]\ne \Lambda_{\mathcal{T}_i}[j])}{B}.
\end{equation}
Hence, $\mathcal {V}$ continues the chained watermark validation process $W_{\mathcal{T}_i}$ by $W_{\mathcal{T}_i}$, checking the extracted watermark detection rate. Once all $W_{\mathcal{T}_i}$ have been successfully verified, the SL model $W$ is deemed to have correctly completed the training process.

\subsection{Security Analysis}
We analyze how \textsc{CliCooper} addresses the identified threats under our serverless, partially trusted, multi-client setting. Internally, rational honest-but-curious trainers seek to infer information from activations or to claim undue credit; externally, adversaries interact only via black-box APIs and attempt unauthorized misuse, e.g., extraction/distillation.

\subsubsection{Threats from internal users (trainers)}
Each trainer $\mathcal{T}_i$ can inspect its parameters and activations it receives, but never observes $\mathcal{C}$’s raw inputs or true labels, peers’ parameters, or the secret label mapping. The principal risks are curiosity-driven inference on activations, including clustering, inversion, membership/attribute/label inference, and credit/ownership disputes, rather than protocol sabotage. 

\begin{packeditemize}
\item \textbf{Privacy-gated label space \& DP activation boundary.}
Trainers operate solely in the pseudo-label space $Y^*$, which removes direct semantic interpretability and decouples observable class frequencies from the true task. At the client boundary, Laplace$(\Delta_1/\varepsilon)$ noise yields a $\varepsilon$-DP interface for smashed activations, limiting linkage between any pseudo-label pairs and restricting reconstruction power.

\item \textbf{Segment opacity \& cross-segment inference hardness.}
Each trainer retains a private model segment and relays only upstream activations; the internal white-box view is local. Without access to peers’ parameters, end-to-end Jacobians, or backward signals beyond the split, standard cross-segment attacks, e.g., reconstructing upstream features or attributing peer capacity/architecture, become ill-posed: the mapping from a single-segment view to others’ weights is non-identifiable, gradient information is locally confined, and multiple vantage points required for triangulation are unavailable by design. This opacity complements the DP boundary, suppressing clustering and inversion reliability.

\item \textbf{Chained watermark lineage for fair attribution and ownership.}
A trainer that reuses a pre-trained/open-source fragment across different tasks could claim rewards without genuine participation. \textsc{CliCooper} prevents such ``free-lunch segment replay'' by chaining watermarks across trainer $\mathcal{T}_{i}$, where each mark $\Lambda_{\mathcal{T}_{i}}$ is derived from $(\!\mathbb{M}_{\mathcal{T}_{i-1}},\;i ,\; \mu_i,\; ID_{\mathcal{T}_{i}})$, binding the segment to the predecessor checkpoint, the run chronology, and the intended task. The fully trusted verifier $\mathcal{V}$ with omniscient audit access reconstructs expected marks and validates detection and utility, establishing segment-level provenance, deterring free-riding, and resolving credit disputes without exposing peers’ parameters.
\end{packeditemize}

\subsubsection{Threats from external adversaries (black-box only)}
Outsiders can query the related model via APIs with no access to architecture or weights, so the primary threat is unauthorized reuse through extraction or functional repurposing.

In \textsc{CliCooper}, outputs reside in a pseudo-label space and are operationally degraded without the secret inverse map $\mathcal{G}_Y^{-1}$, as a black-box extractor or distilled surrogate can at best replicate pseudo semantics, which do not translate to the true task without authorization. This breaks reliable relabeling and task remapping. Recovering $\mathcal{G}_Y^{-1}$ from queries amounts to solving a combinatorial preimage problem under label expansion, further obfuscated by activation-level DP noise that induces stochasticity in outputs and elevates the Bayes risk of any student trained on them.

The black-box release also blocks white-box watermark attacks, including removal, overwriting, pruning, and fine-tuning. Moreover, chained watermarks in \textsc{CliCooper} are derived from hidden checkpoints and identity-bound nonces, so an extracted surrogate cannot reproduce a valid lineage; any redistribution or task-mismatch reuse lacking a consistent chain is flagged by the trusted verifier $\mathcal{V}$ during provenance checks without revealing collaborators’ parameters.

%% file: sections_1/5_experiment.tex
\section{Experiments}\label{sec:experiment}

Our experiments focus on 
\textit{fidelity} (model performance impact), \textit{robustness} (against clustering, inversion, model extraction), and \textit{accountability} (ownership verifiability, overhead).

\subsection{Experimental Settings}
\label{sec:experiment_setting}
Experiments are executed within a consistent micro-server environment for fair comparisons across all configurations.

\noindent\textbf{Setup.}
The experiments are conducted on a high-performance server equipped with an AMD EPYC 9254 CPU operating at 2.9GHz, comprising 24 physical cores and 128MB of L3 cache.
The system is configured with 192GB of 4800MHz ECC DDR5 memory in a twelve-channel layout.
For storage, it provides dual 1.92TB NVMe SSDs. 
The training is accelerated using two NVIDIA L40 GPUs, each offering 18,176 CUDA cores, 568 Tensor Cores, and 48GB of dedicated memory. 

\noindent\textbf{Dataset.} We consider four benchmark datasets spanning both image and text classification tasks. All datasets (including both features and labels) are retained locally at $\mathcal{C}$.

\begin{packeditemize}
    \item \textit{MNIST} is a benchmark dataset for handwritten digit recognition, containing 70,000 grayscale images of digits 0-9, each of size 28$\times$28. It is split into 60,000 training and 10,000 testing samples.
    
    \item \textit{CIFAR-10} comprises 60,000 RGB images of size 32$\times$32 pixels in 10 distinct classes. We adopt 
    50,000 images for training and 10,000 for testing.
    
    \item \textit{CIFAR-100} contains 60,000 32$\times$32 images, but spans 100 fine-grained categories. Each image includes both fine and coarse labels, with the same train-test split as CIFAR-10.
    
    
    \item \textit{AG News} is a text classification corpus containing news headlines and descriptions from four categories. The dataset provides 120,000 training samples and 7,600 test instances, each labelled with its corresponding topic.
    
    \end{packeditemize}

\noindent\textbf{Model architectures.} 
A variety of widely used neural architectures are selected, including \textit{DeepLeNet}, \textit{AlexNet}, \textit{ResNet18}, and \textit{WideResNet} for image classification tasks, and \textit{TextCNN} and \textit{MiniBert} for text understanding. For vision models, \textit{DeepLeNet} and \textit{AlexNet} provide foundational CNN baselines, while \textit{ResNet18} introduces skip connections that facilitate deeper optimization. \textit{WideResNet} enhances accuracy by increasing the network’s width. In the NLP domain, \textit{TextCNN} captures local semantics through convolutional operations, and \textit{MiniBert} delivers efficient contextual encoding through lightweight transformer-based pretraining.

In the SL framework, $\mathcal{C}$ is responsible for initial data embedding and feature projection, after which the activation is forwarded to $\mathcal{T}$ for subsequent training and gradient updates. We apply \textit{DeepLeNet} and \textit{AlexNet} to MNIST, \textit{AlexNet} and \textit{ResNet18} to CIFAR-10, \textit{ResNet18} and \textit{WideResNet} to CIFAR-100, and use both \textit{TextCNN} and \textit{MiniBert} for AG News.

All models are optimized using SGD with momentum set to 0.9. Learning rates are chosen based on the dataset characteristics: 0.01 for MNIST, 0.1 for CIFAR-10 and CIFAR-100, and 0.005 for AG News. A consistent batch size of 256 is used across all training runs to maintain comparability.

\subsection{Model Performance}
\label{sec:model_performance}
To safeguard data and label privacy while enabling training traceability and copyright, \textsc{CliCooper} integrates label expansion, DP noise perturbation, and chained watermarking. Since these mechanisms inevitably interfere with the SL training, it is necessary to assess their impact on model performance. We therefore conducted three sets of experiments: (1) fixing $\gamma$=2.0 and $\varepsilon$=5.0, while varying watermark lengths $B\in \{512, 1024, 2048\}$; (2) fixing $\varepsilon$=5.0 and $B$=1024, while varying $\gamma\in$ {1.5, 2.0, 2.5}; and (3) fixing $B$=1024 and $\gamma$=2.0, while varying $\varepsilon \in$ {2.0, 5.0, 10.0}. The baseline is a standard SL model without any protection, achieving accuracies of 99.42\% (DeepLeNet on MNIST), 99.45\% (AlexNet on MNIST), 87.97\% (AlexNet on CIFAR-10), 92.32\% (ResNet18 on CIFAR-10), 71.97\% (ResNet18 on CIFAR-100), 71.14\% (WideResNet on CIFAR-100), 86.59\% (TextCNN on AG News), and 90.68\% (MiniBert on AG News).

Figs.~\ref{Watermark_Acc}--\ref{Epsilon_Acc} indicate that, in most cases, \textsc{CliCooper} has a negligible impact on accuracy, with performance remaining close to the baseline. This can be attributed to the redundancy of deep network parameters and the existence of multiple local optima; watermark embedding or DP-based noise injection guides the model toward different local optima without degrading accuracy. 
{\color{black}Label expansion does not reduce $Acc_{main}$ because we evaluate in the true-label space: predictions over pseudo-labels are mapped back to the original labels using the secret mapping, so predicting a different pseudo-label within the same true-label group is not counted as an error. Moreover, label expansion is accompanied by data augmentation to avoid data sparsity as the pseudo-label space grows.}

In some extreme settings, $\varepsilon=2.0$, accuracy drops by up to 3\% due to excessive noise interference. Interestingly, in certain cases, noise injection even improves generalization: with $\varepsilon=10.0$, AlexNet on CIFAR-10 exhibits a 1\% accuracy gain, while ResNet18 on CIFAR-100 improves by 1.2\%, which suggests that an appropriate noise can act as a regularizer, alleviating overfitting and slightly enhancing performance.


\begin{figure*}[t]
    \centering
    \subfigure[AlexNet\_MNIST]{
    \begin{minipage}[t]{0.23\textwidth}
    \centering
    \includegraphics[width=1.35in]{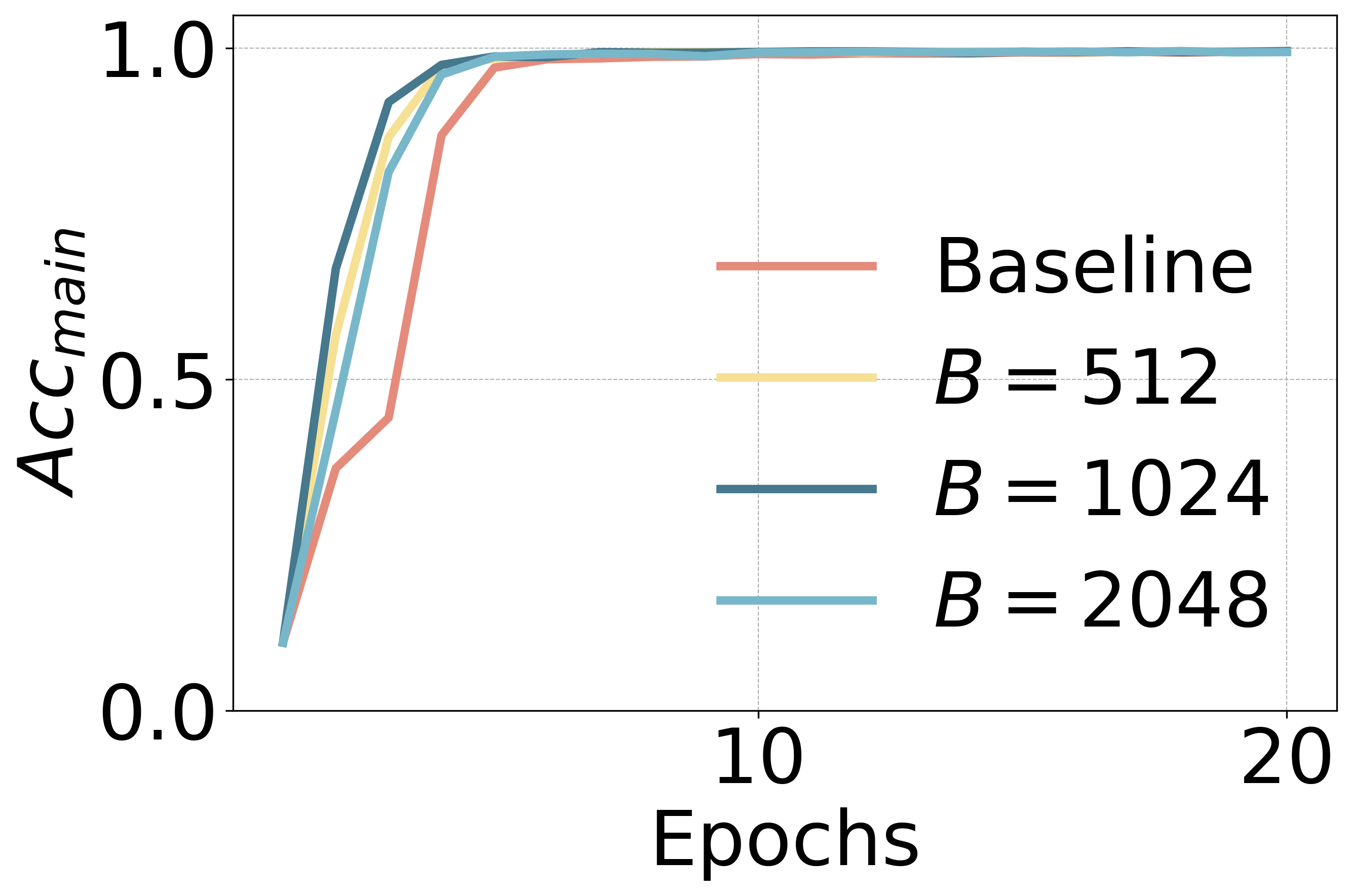}
    \end{minipage}
    \label{resnet_mnist_wm_acc}
    }
    \subfigure[ResNet18\_CIFAR10]{
    \begin{minipage}[t]{0.23\textwidth}
    \centering
    \includegraphics[width=1.35in]{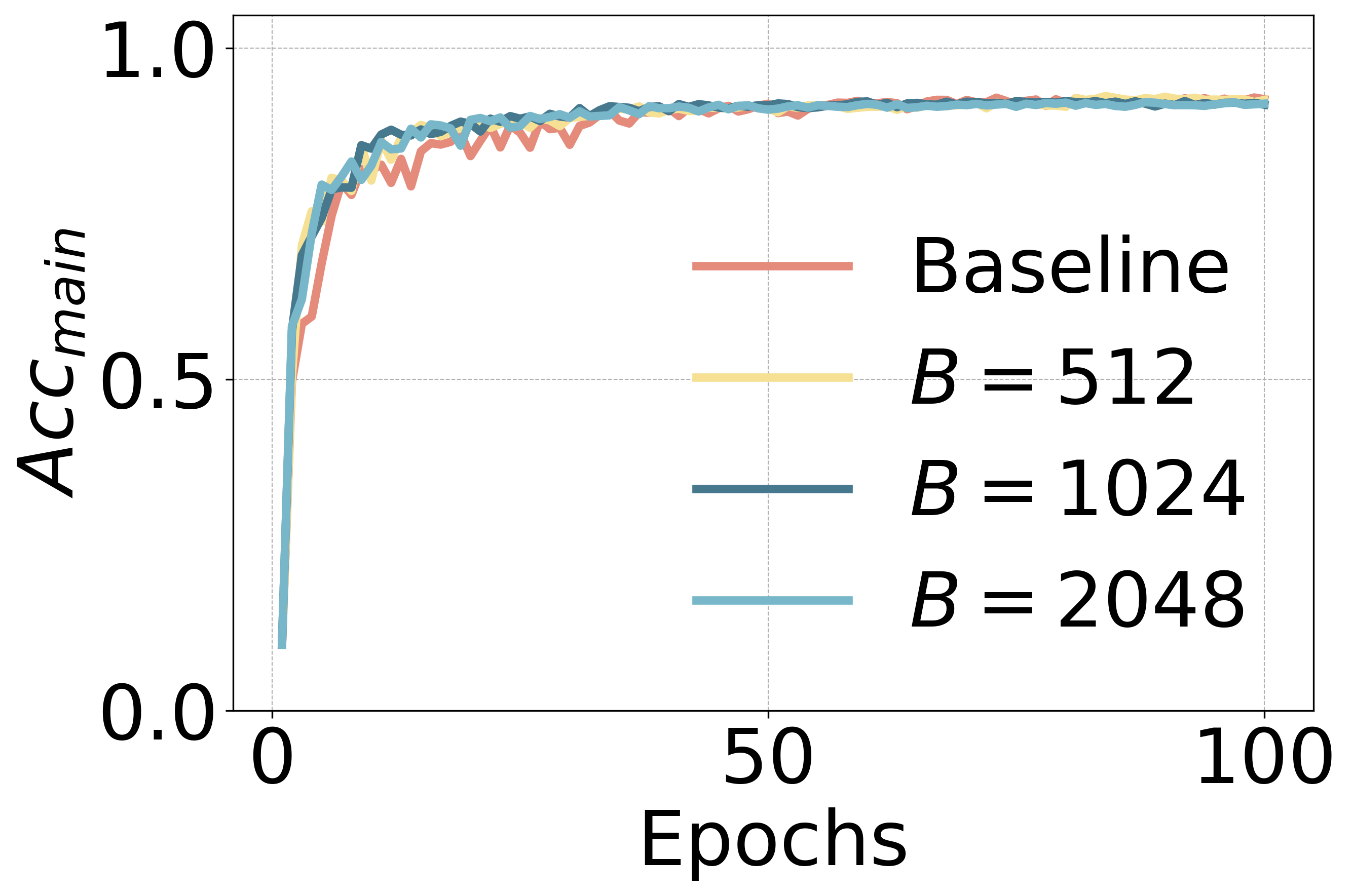}
    \end{minipage}
    \label{resnet18_cifar10_wm_acc}
    }
    \subfigure[WideResNet\_CIFAR100]{
    \begin{minipage}[t]{0.23\textwidth}
    \centering
    \includegraphics[width=1.35in]{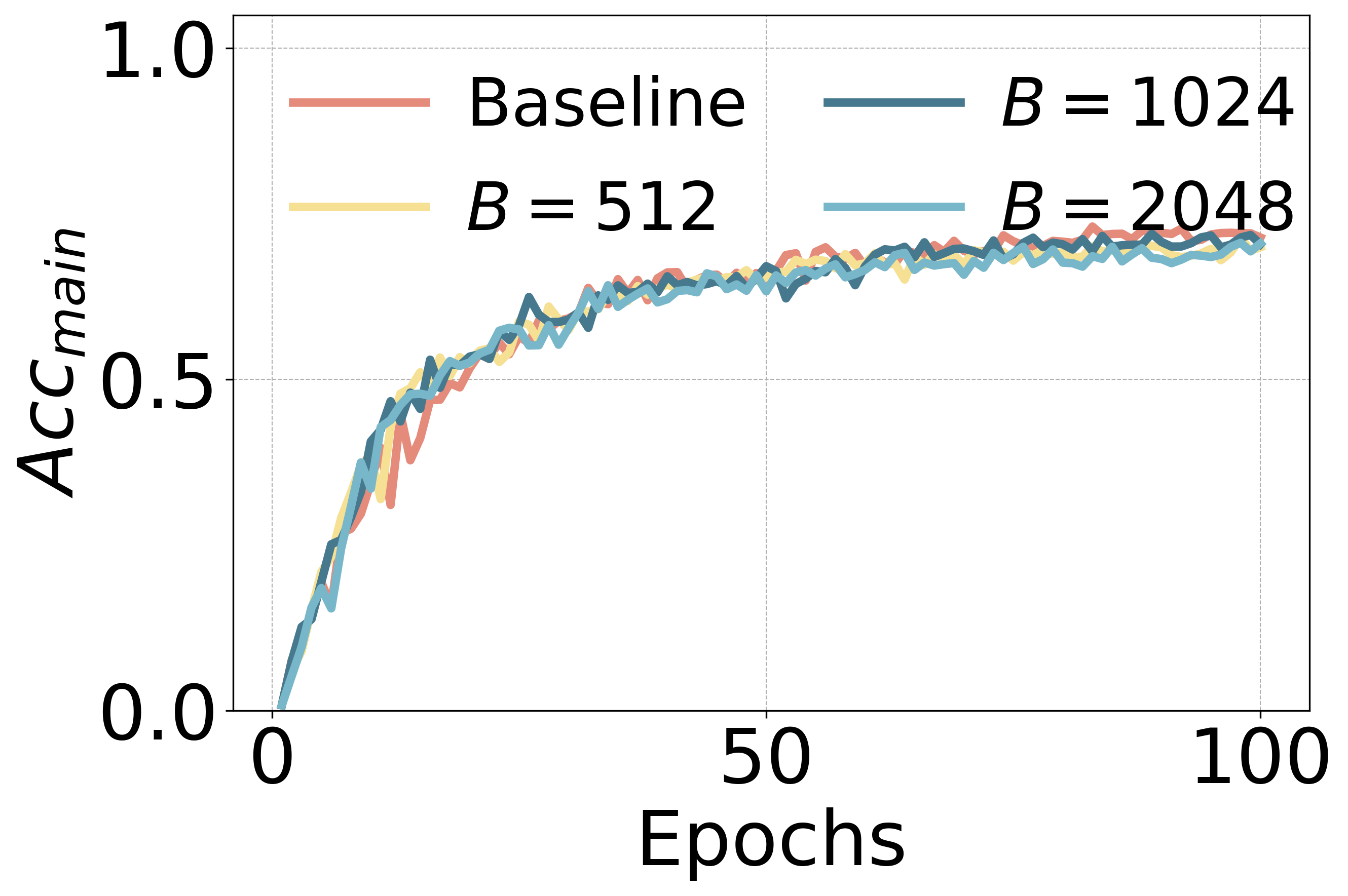}
    \end{minipage}
    \label{wideresnet_cifar100_wm_acc}
    }
    \subfigure[MiniBert\_AG News]{
    \begin{minipage}[t]{0.23\textwidth}
    \centering
    \includegraphics[width=1.35in]{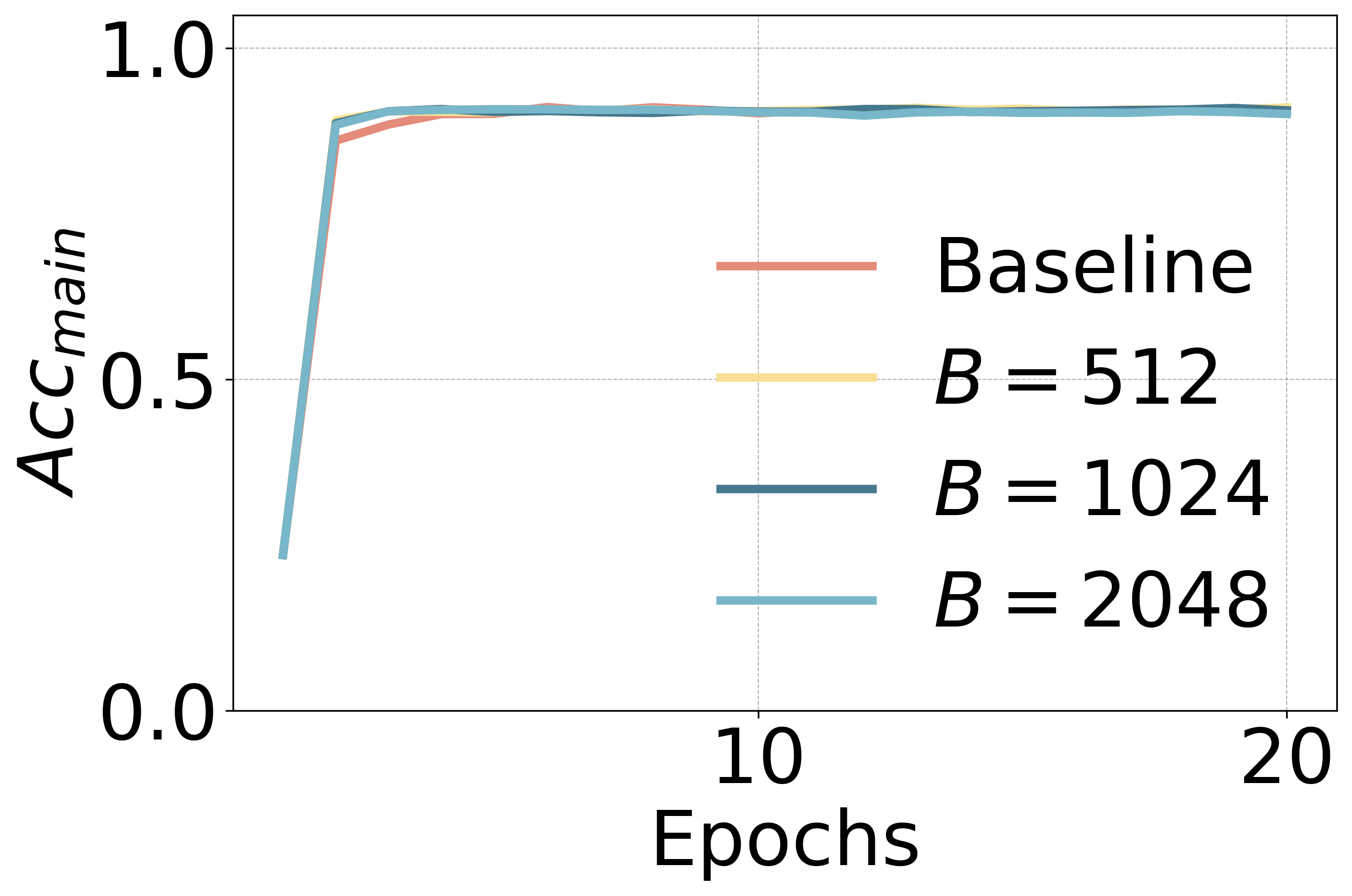}
    \end{minipage}
    \label{bert_news_wm_acc}
    }
\caption{
Comparison of the main accuracy with baseline across different architectures and datasets, where $\varepsilon$=5.0, $\gamma$=2.0, $B$=512/1024/2048.
}
\label{Watermark_Acc}
\end{figure*}

\begin{figure*}[t]
    \centering
    \subfigure[AlexNet\_MNIST]{
    \begin{minipage}[t]{0.23\textwidth}
    \centering
    \includegraphics[width=1.35in]{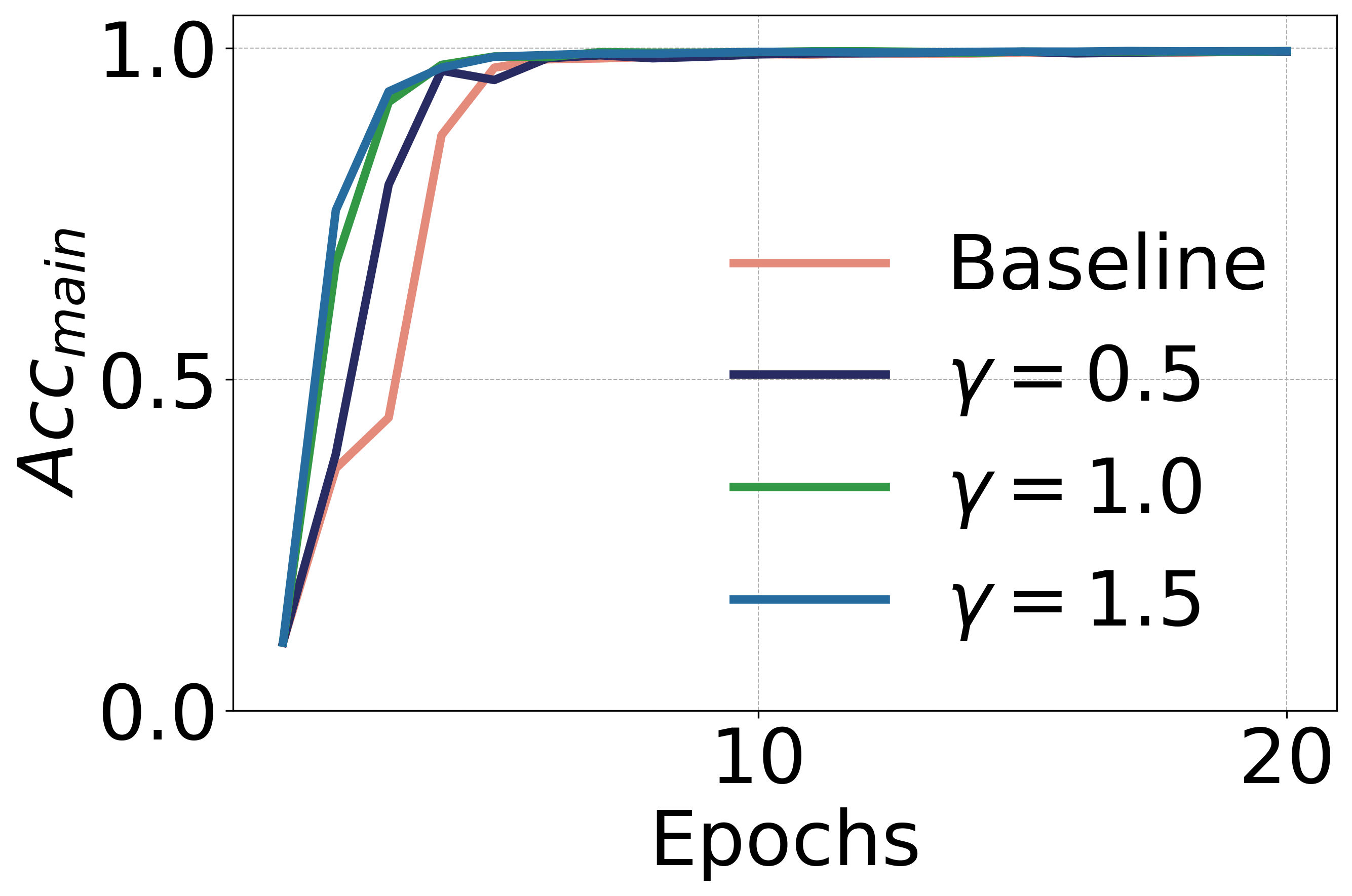}
    \end{minipage}
    \label{alexnet_mnist_rate_acc}
    }
    \subfigure[ResNet18\_CIFAR10]{
    \begin{minipage}[t]{0.23\textwidth}
    \centering
    \includegraphics[width=1.35in]{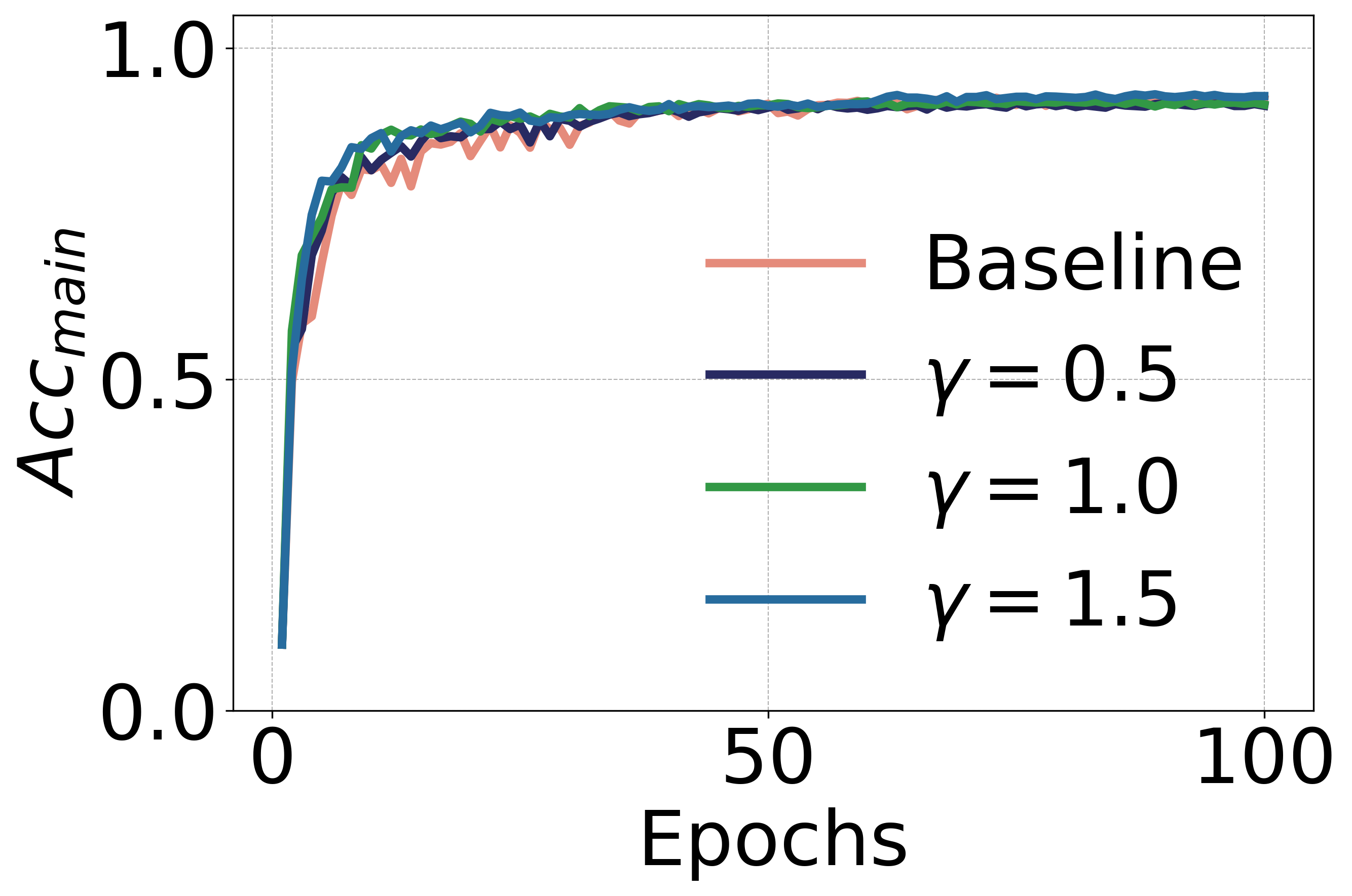}
    \end{minipage}
    \label{resnet18_cifar10_rate_acc}
    }
    \subfigure[WideResNet\_CIFAR100]{
    \begin{minipage}[t]{0.23\textwidth}
    \centering
    \includegraphics[width=1.35in]{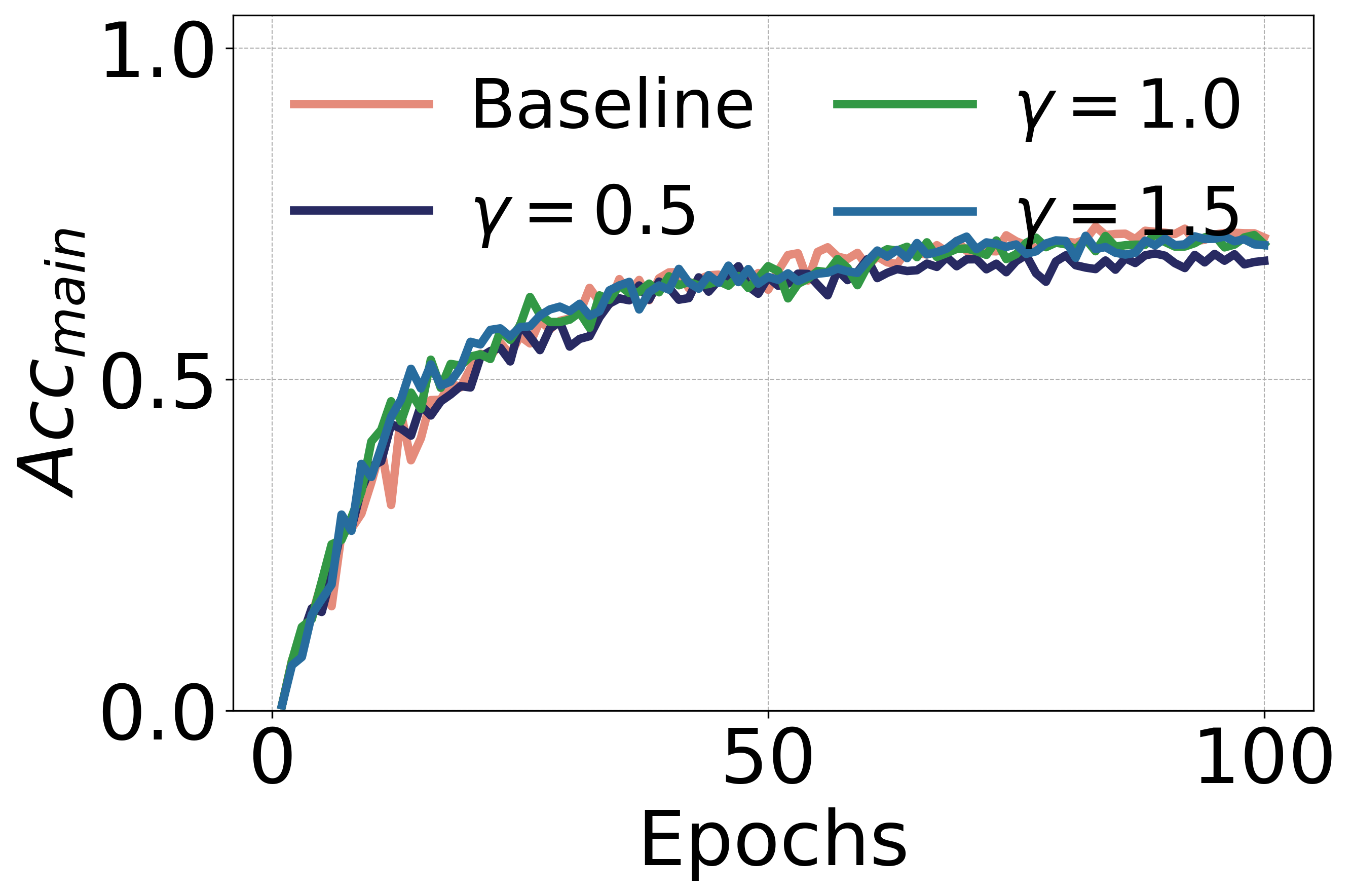}
    \end{minipage}
    \label{wideresnet_cifar100_rate_acc}
    }
    \subfigure[MiniBert\_AG News]{
    \begin{minipage}[t]{0.23\textwidth}
    \centering
    \includegraphics[width=1.35in]{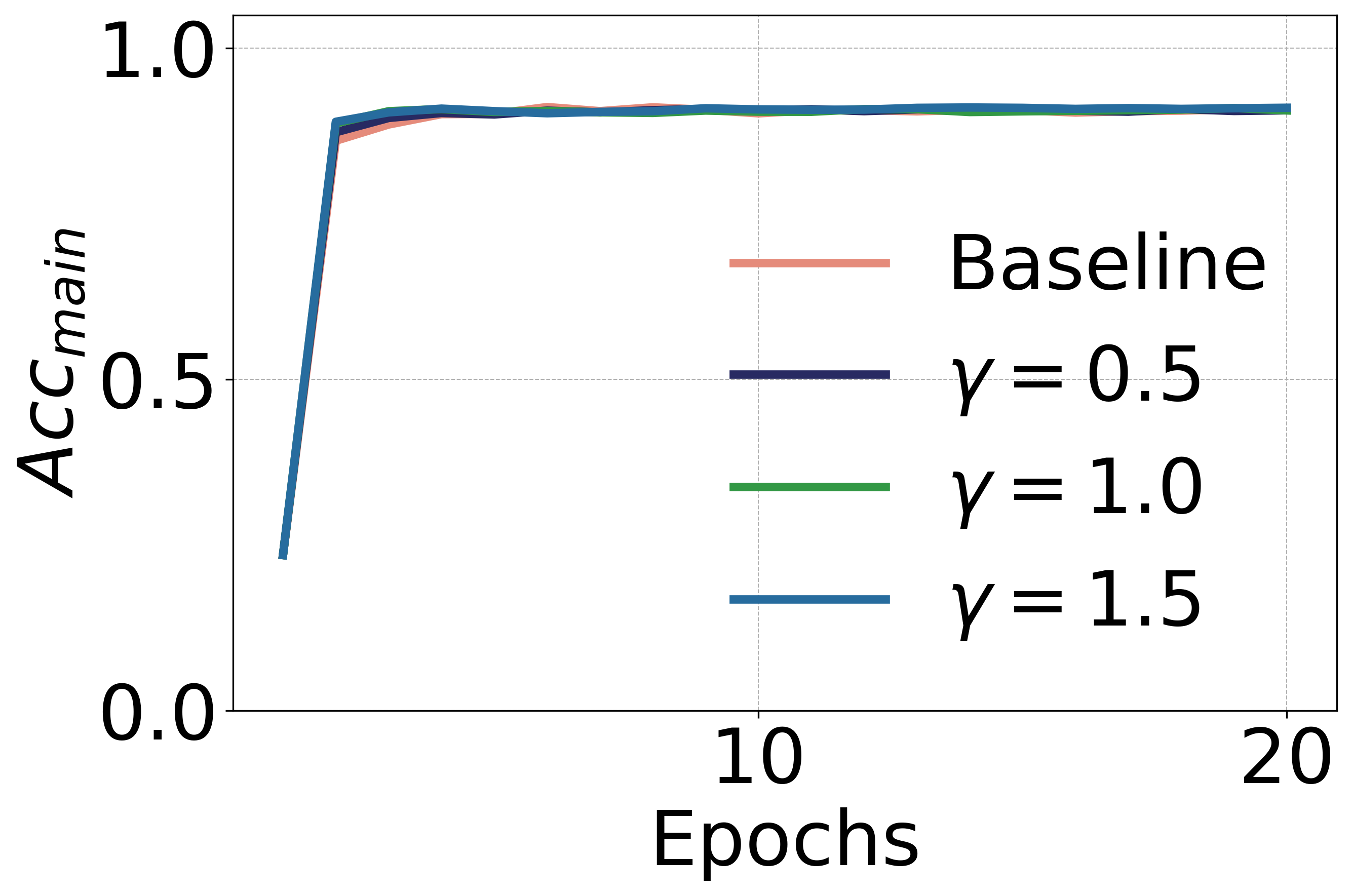}
    \end{minipage}
    \label{bert_news_rate_acc}
    }
\caption{
Comparison of the main accuracy with baseline across different architectures and datasets, where $\varepsilon$=5.0, $B$=1024, $\gamma$=1.5/2.0/2.5.
}
\label{Rate_Acc}
\end{figure*}

\begin{figure*}[t]
    \centering
    \subfigure[AlexNet\_MNIST]{
    \begin{minipage}[t]{0.23\textwidth}
    \centering
    \includegraphics[width=1.35in]{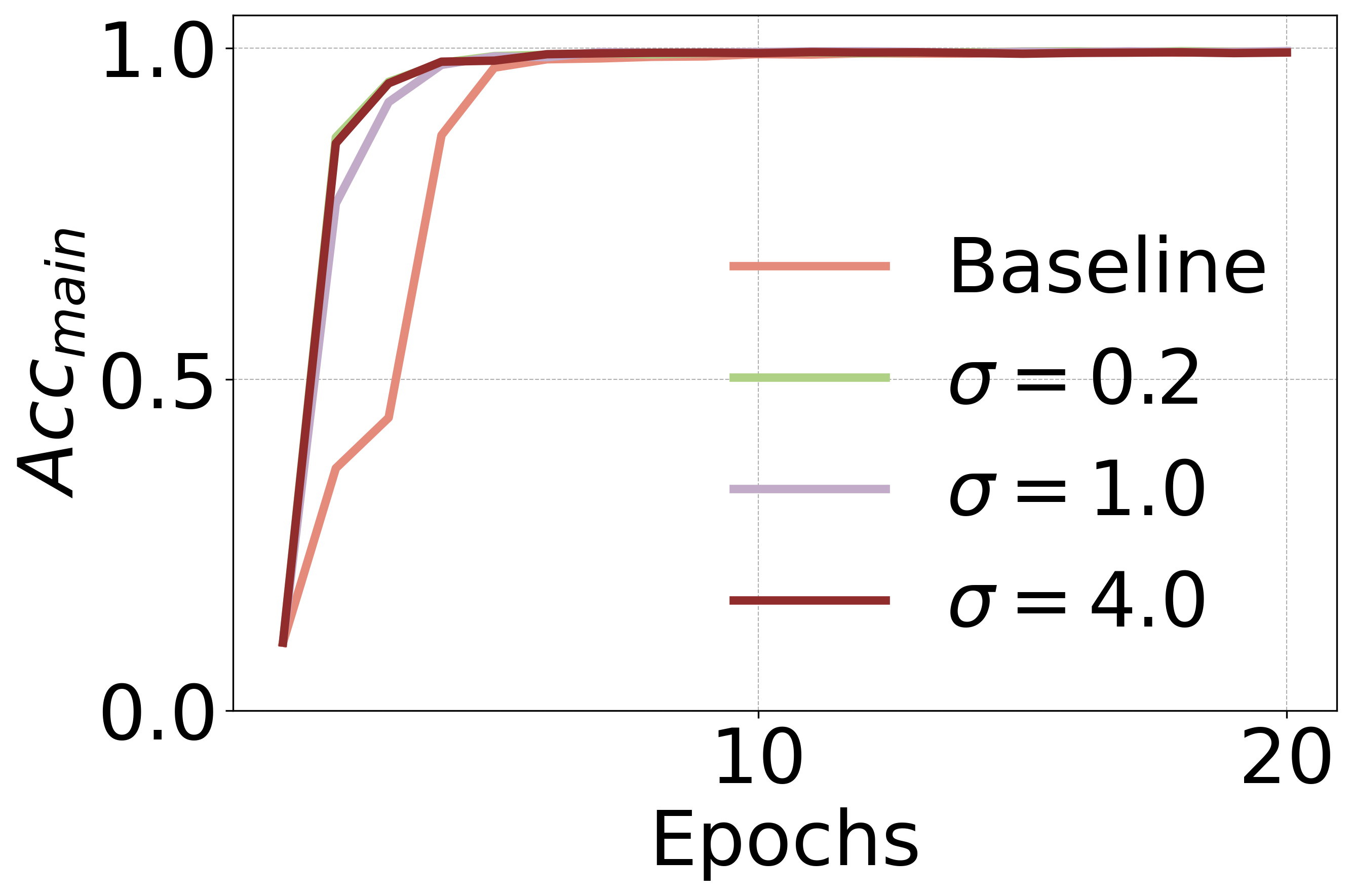}
    \end{minipage}
    \label{alexnet_mnist_epsilon_acc}
    }
    \subfigure[ResNet18\_CIFAR10]{
    \begin{minipage}[t]{0.23\textwidth}
    \centering
    \includegraphics[width=1.35in]{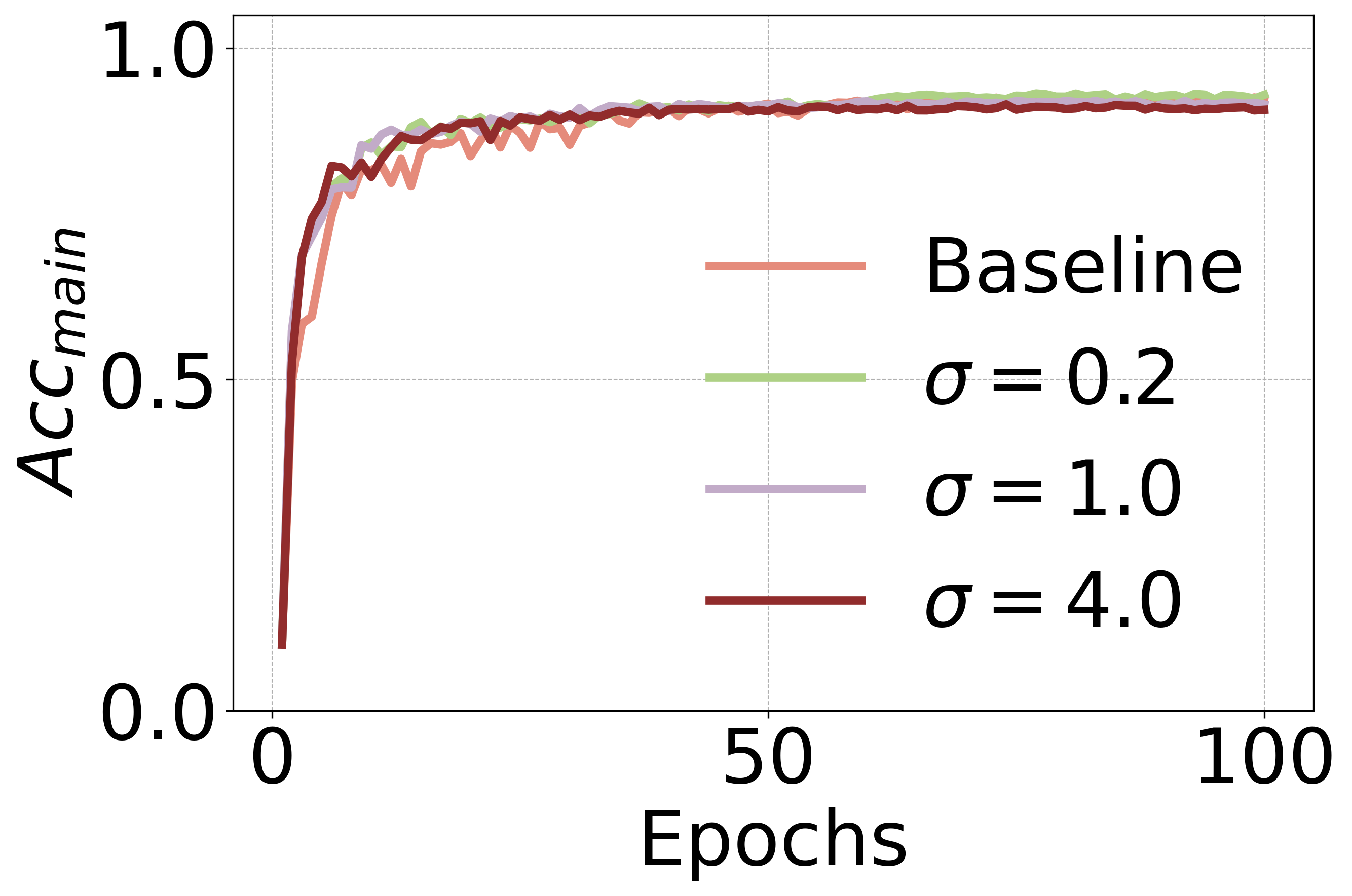}
    \end{minipage}
    \label{resnet18_cifar10_epsilon_acc}
    }
    \subfigure[WideResNet\_CIFAR100]{
    \begin{minipage}[t]{0.23\textwidth}
    \centering
    \includegraphics[width=1.35in]{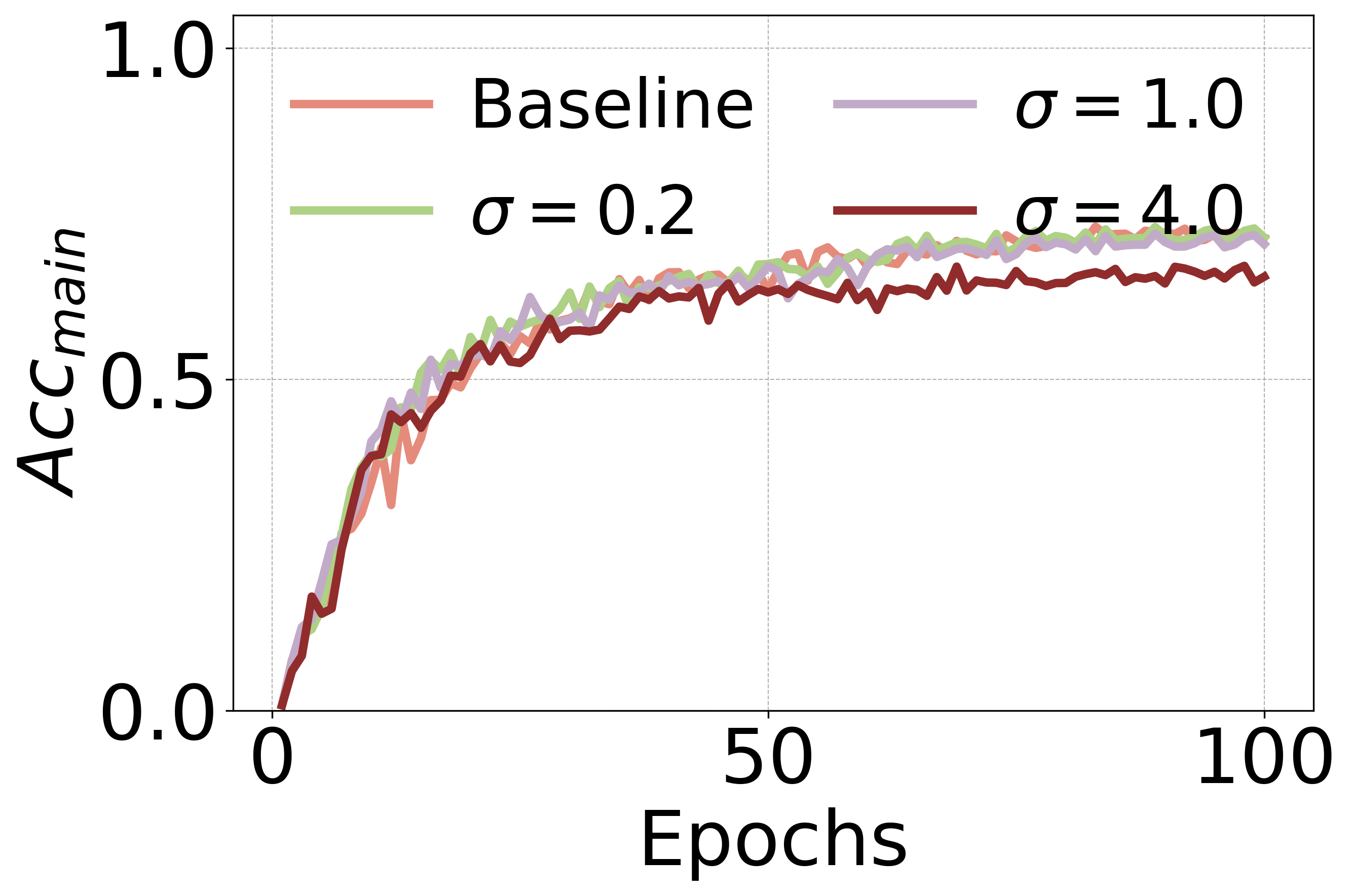}
    \end{minipage}
    \label{wideresnet_cifar100_epsilon_acc}
    }
    \subfigure[MiniBert\_AG News]{
    \begin{minipage}[t]{0.23\textwidth}
    \centering
    \includegraphics[width=1.35in]{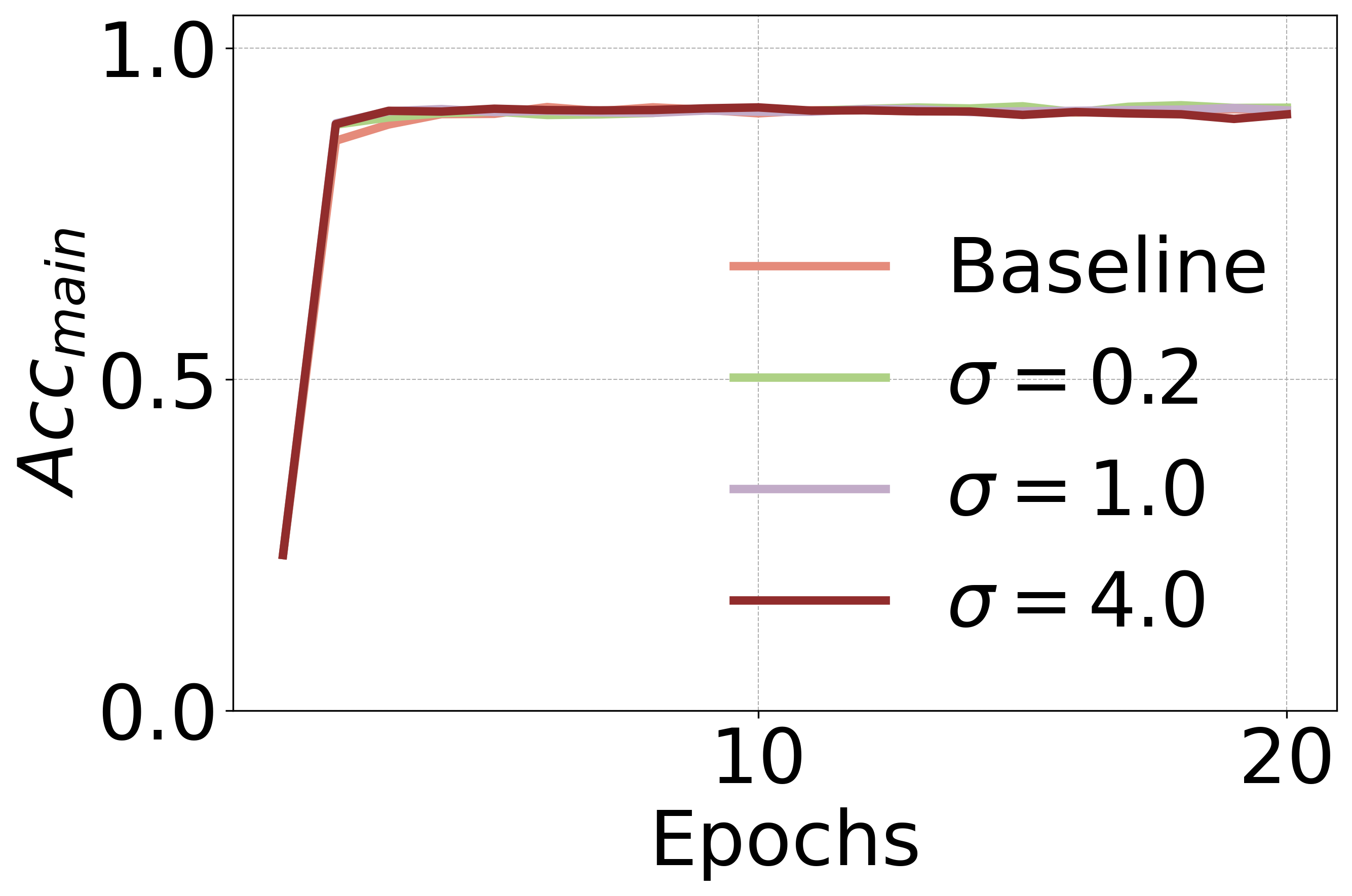}
    \end{minipage}
    \label{bert_news_epsilon_acc}
    }
\caption{
Comparison of the main accuracy with baseline across different architectures and datasets, where $B$=1024, $\gamma$=2.0, $\varepsilon$=2.0/5.0/10.0.
}
\label{Epsilon_Acc}
\end{figure*}

\subsection{Resistance to Data Clustering Attacks}
\label{sec:cluster_attack}
\input{tables/cluster-perfect}

\begin{table*}[!ht]
\centering
\caption{
Visual comparison of inversion reconstructed samples under DP noise levels ($\varepsilon$) for CIFAR-10 and CIFAR-100.}
\renewcommand{\arraystretch}{0.9}
\setlength{\tabcolsep}{8pt} 
\resizebox{\textwidth}{!}{
\begin{threeparttable}
{\Large
\begin{tabular}{ccc}
\hline
       & CIFAR-10 & CIFAR-100 \\ \hline
Target & \includegraphics[width=\linewidth]{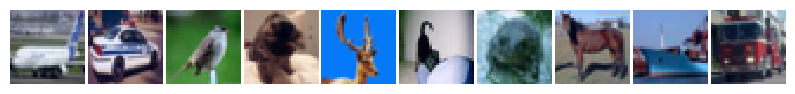}     & \includegraphics[width=\linewidth]{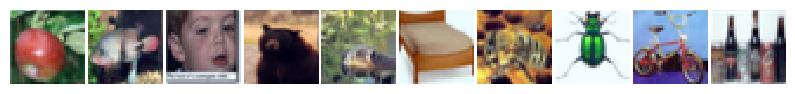}        \\ \hline
Baseline   & \includegraphics[width=\linewidth]{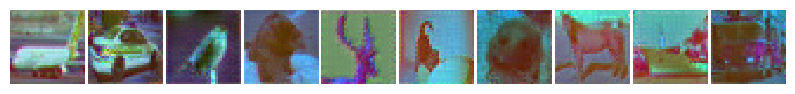}     & \includegraphics[width=\linewidth]{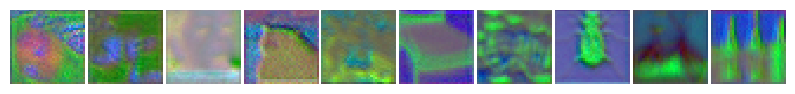}       \\ \hline
$\varepsilon$=10.0    & \includegraphics[width=\linewidth]{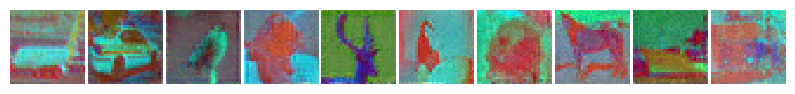}     & \includegraphics[width=\linewidth]{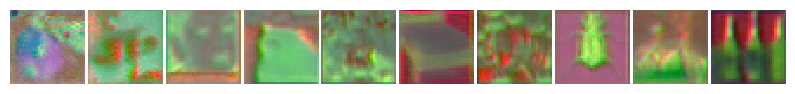}       \\ \hline
$\varepsilon$=5.0    & \includegraphics[width=\linewidth]{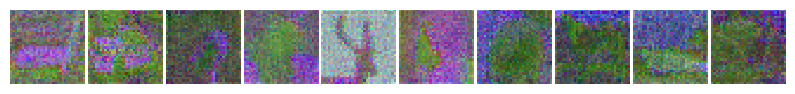}     & \includegraphics[width=\linewidth]{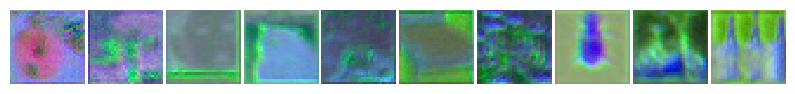}       \\ \hline
$\varepsilon$=2.0   & \includegraphics[width=\linewidth]{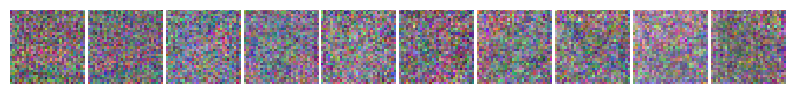}     & \includegraphics[width=\linewidth]{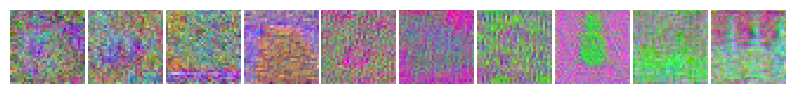}       \\ \hline
\end{tabular}
}
\label{inversion_attack}
\vspace{4pt}
\begin{tablenotes} 
{\Large
\vspace{1ex}
\item $\bullet$Reconstructing from unprotected activation yields relatively high similarity to the input sample—SSIM 0.50 on CIFAR-10 and 0.26 on CIFAR-100. Injecting DP noise at release markedly reduces fidelity: on CIFAR-10, SSIM drops to 0.38, 0.22, and 0.03; on CIFAR-100, it falls to 0.20, 0.15, and 0.07 for $\varepsilon$=10.0, 5.0, 2.0, respectively.
}
\end{tablenotes}
\end{threeparttable}
}
\end{table*}

{\color{black} We evaluate a clustering-based label matching attacker in the split-learning setting where an internal trainer $\mathcal{T}_1$ cannot access the client’s raw inputs or true labels, but can observe DP-protected activations and the corresponding pseudo-labels produced by label expansion. The attacker applies \textbf{K-means}~\cite{clustering_b}, \textbf{Birch}~\cite{birch}, and \textbf{DBSCAN}~\cite{dbscan} to the observed activations (all the clustering methods infer the number of clusters automatically) and tries to recover the hidden grouping structure induced by label expansion. We run the attack three times and report the average \textbf{Perfect Clustering Accuracy}, defined as the fraction of true-label groups recovered exactly (i.e., each predicted cluster matches a complete true group with no extra or missing samples).
    
}

{\color{black}Table~\ref{table-perfect-cluster} shows that without protection (Baseline), clustering accuracy is 100\% because activations and true labels are one-to-one, making grouping trivial. Under the protection, on CIFAR10/CIFAR100, clustering fails: clustering accuracy is 0\% across settings; i.e., DP perturbation plus label expansion largely destroys cluster separability in harder vision tasks. On MNIST, the attacker achieves moderate success ($\le $50\%) because the underlying representations are more separable, making residual structure easier to exploit even after perturbation. The easiest case is AG News, where clustering accuracy reaches 50\%–90\%; text representations can retain strong class separability and thus remain vulnerable to unsupervised grouping in some configurations.

As expected, decreasing $\varepsilon$ (stronger DP noise) reduces both accuracies by obscuring activation geometry, while increasing $\gamma$ (more pseudo-labels per true class) further lowers accuracy by increasing intra-class dispersion and inter-group mixing, making correct grouping harder for unsupervised clustering.}
Label expansion plus DP noise renders $\mathbb{M}{\mathcal{C}}^{\text{\tiny DP}}$ non-clusterable to an honest-but-curious $\mathcal{T}_1$, addressing \textcolor{violet}{\textbf{RQ1} (data privacy)}.

\subsection{Resistance to Data Inversion Attacks}
\label{sec:inversion_attack}

An honest-but-curious $\mathcal{T}_1$, given only the activation uploaded by $\mathcal{C}$ and no access to $\mathcal{C}$’s parameters, may attempt to reconstruct $\mathcal{C}$’s training data. Following the UNSplit framework~\cite{unsplit_27}, $\mathcal{T}_1$ initializes a surrogate network matching $\mathcal{C}$’s architecture and poses reconstruction as joint optimization over inputs and surrogate weights. We evaluate robustness on two representative settings—ResNet18/CIFAR-10 and WideResNet/CIFAR-100—where $\mathcal{C}$ performs training-data embedding and applies DP to the activation.

We quantify similarity using Structural Similarity Index (SSIM) and summarize visuals in Table~\ref{inversion_attack}. The first row contains target images; the second shows reconstructions from the unprotected baseline, which recover clear shapes and textures, indicating severe leakage. Introducing DP progressively degrades reconstructions: with $\varepsilon=10.0$, outlines and colors blur; at $\varepsilon=5.0$, structures largely vanish, and noise dominates; by $\varepsilon=2.0$, outputs collapse into snowflake-like patterns devoid of semantic content. SSIM scores follow the same monotonic decline.
Thus, stronger privacy budgets (smaller $\varepsilon$) render UNSplit ineffective in practice, while label expansion prevents alignment between any residual visual cues and true class semantics, adding defense-in-depth. The label expansion and DP noise on activations thwart inversion attempts, ensuring that \textsc{CliCooper} addresses \textcolor{violet}{\textbf{RQ1} (data privacy)}.

\subsection{Resistance to Model Extraction Attacks}
\label{sec:extraction_attack}

In our setting, the trained SL model is released in a black-box manner: adversaries cannot access the architecture or parameters and can only issue queries to the public API. Given an unlabeled dataset, an attacker may attempt model extraction~\cite{model_extraction_35}
by querying the API for pseudo-labels and training a surrogate. Our label expansion prevents this: the API returns only pseudo-labels, and without the one-to-many and inverse mappings 
($\mathcal{G}_Y$,$\mathcal{G}_{Y}^{-1}$), the true labels are unrecoverable.

To test robustness, we construct unlabeled inputs by adding noise to MNIST, CIFAR10, CIFAR100 training data and by randomly shuffling tokens for AG News. We query the API for pseudo-labels with expansion factor $\gamma=2.0$ (e.g., 10 true classes $\to $ 20 pseudo-labels). Lacking ($\mathcal{G}_Y$,$\mathcal{G}_{Y}^{-1}$), the attacker can only map each pseudo-label group to a true class at random, then trains a surrogate on the pseudo-labeled set and evaluates it on the true test set.

{\color{black}Table~\ref{SM_Acc} summarizes outcomes: the second to fourth columns show the SL model $W$ accuracy, the pseudo-labeled dataset $\mathbb{D}'$ accuracy, and the attacker’s surrogate model $W'$ accuracy, respectively. Across datasets/architectures, attacker performance remains near random guess, $\approx $10\% for MNIST and CIFAR10, 1\% for CIFAR100, 25\% for AG News, far below the SL model (e.g., 99.5\% on MNIST, $\approx $90\% on AG News). Thus, API outputs cannot be exploited to assemble a meaningful training set or a competitive surrogate.}
To this end, \textsc{CliCooper} resolves \textcolor{violet}{\textbf{RQ3} (unauthorized-use defense)}.

\input{tables/model-extraction}

\input{tables/time_acc}

\input{tables/communication-T}

\subsection{Ownership Verifiability and Overhead}

As shown in Table~\ref{time-acc}, watermark extraction achieves over 99\% accuracy across all models and settings, providing reliable and automated proof of contribution. This detection rate confirms the effectiveness of the chained design, where each trainer’s seed derives from its predecessor’s activation, producing stable model-dependent signals that cannot be precomputed or substituted without following the designated training process. 
Increasing watermark length $B$ predictably raises embedding and verification time (linear encoder/decoder), yet the absolute costs remain small relative to training: even at $B=2048$, embed and verify each stay in the single– to low–tens–of–milliseconds range, so the end-to-end impact is marginal and does not affect scheduling or throughput for compensated training.
Although total training time increases with the label-expansion factor $\gamma$ due to larger pseudo-labeled datasets, this effect is independent of watermarking and remains within acceptable bounds.
%

{\color{black}Table~\ref{table-communicatoin} reports the per-link communication latency in \textsc{CliCooper} with a batch size of 256 and an assumed bandwidth of 200 MB/s. Since only intermediate activations are relayed along  $\mathcal{C} \rightarrow \mathcal{T}_1 \rightarrow \mathcal{T}_2 \rightarrow \mathcal{T}_3$, the latency mainly depends on the interface tensor size and thus increases with model/dataset complexity.
For $n$ trainer clients in a relay chain, the per-iteration communication time scales linearly with the number of segment boundaries and the interface size, i.e., $T_{comm}=O(n(\frac{\bar{s}}{R}+\theta ))$, where $\bar{s}$ is the average interface tensor size, $R$ is the link bandwidth, and $\theta$ is the per-link propagation/handshake latency.
The worst case is CIFAR100 with WideResNet, where the bottleneck link takes 0.80 s, and the total end-to-end latency is 1.40 s. In contrast, AG News with TextCNN yields a bottleneck of 0.06 s and a total of 0.10 s. 
Even the maximum latency is negligible compared to the training time in Table~\ref{time-acc}; multi-trainer coordination incurs only modest overhead under typical bandwidth.
}
To this end, \textsc{CliCooper} resolves \textcolor{violet}{\textbf{RQ2} (layer ownership)}.

\subsection{\color{black}Discussion on Hyperparameter Choices}
{\color{black}
\noindent\textbf{Average expansion factor $\gamma$}. $\gamma$ controls the average label expansion ratio: a larger $\gamma$ means each true label is mapped to more pseudo labels.
Empirically, Fig.~\ref{Rate_Acc} shows that changing $\gamma$ does not degrade the main-task accuracy nor slow down convergence.
Meanwhile, Table~\ref{table-perfect-cluster} indicates a clear privacy trend: smaller $\gamma$ yields higher clustering-attack accuracy, because fewer pseudo labels per true class leave stronger and more compact grouping structures for attackers to exploit. 
In terms of efficiency, Table~\ref{time-acc} shows that training time increases with $\gamma$, since a larger pseudo-label space requires more augmented samples, which enlarges the training set and computational cost. 
Overall, $\gamma$ introduces a privacy–efficiency trade-off while keeping utility stable; we recommend $\gamma$=2.0 as a practical default balancing stronger protection and moderate overhead.

\noindent\textbf{DP budget $\varepsilon$}. $\varepsilon$ determines the strength of DP noise injected into activations: a smaller $\varepsilon$ corresponds to a stronger perturbation. 
Fig.~\ref{Epsilon_Acc} suggests that overly strong perturbation (e.g., $\varepsilon$=2.0) can incur up to about 3\% accuracy loss in some settings, indicating that excessive noise may distort useful representations. On the other hand, Table~\ref{table-perfect-cluster} shows that reducing $\varepsilon$ suppresses clustering-based label matching, and Table~\ref{inversion_attack} further confirms that stronger DP noise leads to lower SSIM under inversion attacks, i.e., weaker reconstruction quality. Therefore, $\varepsilon$ governs a direct privacy–utility trade-off. 
We select $\varepsilon$=5.0 as a reasonable compromise that substantially reduces attack success while keeping accuracy close to baseline.

\noindent\textbf{Watermark length $B$}. $B$'s length affects the payload capacity and robustness of trainer identification, but it should not penalize utility. 
Fig.~\ref{Watermark_Acc} shows that varying $B$ does not noticeably change model accuracy, consistent with watermark embedding applied after convergence and behaving as a lightweight post-training step. 
Although Table~\ref{time-acc} indicates that larger $B$ increases embedding and verification time, the overhead remains millisecond-level, negligible compared to end-to-end training time. 
Hence, $B$ can be chosen based on the required identity/traceability information (e.g., longer $B$ for richer identifiers), without materially impacting training.
}

%% file: tables/cluster-perfect.tex
\begin{table*}[!t]
\centering
\caption{\color{black}Perfect clustering accuracy (in \%) under different cluster methods and protection settings.}
\label{table-perfect-cluster}
\renewcommand{\arraystretch}{0.92}

\resizebox{\textwidth}{!}{
\begin{threeparttable}
\begin{tabular}{clllllllll}
\toprule
\multicolumn{2}{c}{Cluster Methods}          & \makecell{MNIST \\ DeepLeNet} & \makecell{MNIST \\ AlexNet} & \makecell{CIFAR10 \\ AlexNet} & \makecell{CIFAR10 \\ ResNet18} & \makecell{CIFAR100 \\ ResNet18} & \makecell{CIFAR100 \\ WideResNet} & \makecell{AG News \\ TextCNN} & \makecell{AG News \\ MiniBert} \\ \hline
\multicolumn{2}{c}{Baseline}                 & 100.00            & 100.00          & 100.00           & 100.00            & 100.00             & 100.00               & 100.00           & 100.00            \\ \hline
\multirow{6}{*}{K-means}  & $\varepsilon$=2.0  & 23.33             & 16.66           & 0.00             & 0.00              & 0.00               & 0.00                 & 41.66            & 83.33             \\
                        & $\varepsilon$=5.0  & 40.00             & 30.00           & 0.00             & 0.00              & 0.00               & 0.00                 & 50.00            & 91.66             \\
                        & $\varepsilon$=10.0 & 46.66             & 50.00           & 0.00             & 0.00              & 0.00               & 0.00                 & 66.66            & 91.66             \\
                        & $\gamma$=1.5       & 23.33             & 20.00           & 0.00             & 0.00              & 0.00               & 0.00                 & 75.00            & 91.66             \\
                        & $\gamma$=2.0       & 20.00             & 13.33           & 0.00             & 0.00              & 0.00               & 0.00                 & 50.00            & 91.66             \\
                        & $\gamma$=2.5       & 13.33             & 13.33           & 0.00             & 0.00              & 0.00               & 0.00                 & 33.33            & 83.33             \\ \hline
\multirow{6}{*}{Birch}  & $\varepsilon$=2.0  & 13.33             & 13.33           & 0.00             & 0.00              & 0.00               & 0.00                 & 16.66            & 91.66             \\
                        & $\varepsilon$=5.0  & 30.00             & 16.66           & 0.00             & 0.00              & 0.00               & 0.00                 & 33.33            & 91.66             \\
                        & $\varepsilon$=10.0 & 43.33             & 60.00           & 0.00             & 0.00              & 0.00               & 0.00                 & 75.00            & 91.66             \\
                        & $\gamma$=1.5       & 56.66             & 50.00           & 0.00             & 0.00              & 0.00               & 0.00                 & 66.66            & 91.66             \\
                        & $\gamma$=2.0       & 53.33             & 46.66           & 0.00             & 0.00              & 0.00               & 0.00                 & 50.00            & 91.66             \\
                        & $\gamma$=2.5       & 26.66             & 26.66           & 0.00             & 0.00              & 0.00               & 0.00                 & 26.66            & 83.33             \\ \hline
\multirow{6}{*}{DBSCAN} & $\varepsilon$=2.0  & 0.00              & 3.33            & 0.00             & 0.00              & 0.00               & 0.00                 & 0.00             & 0.00              \\
                        & $\varepsilon$=5.0  & 10.00             & 10.00           & 0.00             & 0.00              & 0.00               & 0.00                 & 0.00             & 8.33              \\
                        & $\varepsilon$=10.0 & 13.33             & 13.33           & 0.00             & 0.00              & 0.00               & 0.00                 & 0.00             & 8.33              \\
                        & $\gamma$=1.5       & 10.00             & 16.66           & 0.00             & 0.00              & 0.00               & 0.00                 & 25.00             & 16.67              \\
                        & $\gamma$=2.0       & 6.66              & 13.33           & 0.00             & 0.00              & 0.00               & 0.00                 & 16.67            & 0.00              \\
                        & $\gamma$=2.5       & 6.66             & 13.33           & 0.00             & 0.00              & 0.00               & 0.00                 & 0.00            & 0.00             \\ \bottomrule
\end{tabular}
\end{threeparttable}
}
\end{table*}

%% file: tables/model-extraction.tex
\begin{table}[]
\centering
\caption{
Accuracy (\%) for resisting model extraction attacks.}
\label{SM_Acc}
\renewcommand{\arraystretch}{0.9}
\begin{threeparttable}
\begin{tabular}{lccc}
\toprule
\multicolumn{1}{c}{Datasets/Architectures} & $ Acc_{W} $ & $ Acc_{\mathcal{\mathbb{D}'}} $ & $Acc_{W'} $  \\ 
\midrule
MNIST DeepLeNet             & 99.53 & 9.87   & 10.51  \\ 
MNIST AlexNet               & 99.57 & 9.65   & 14.44  \\ 
CIFAR10 AlexNet              & 88.38 & 10.34  & 14.31  \\ 
CIFAR10 ResNet18             & 92.23 & 13.63  & 13.14  \\ 
CIFAR100 ResNet18            & 72.27 & 1.30   & 1.44  \\ 
CIFAR100 WideResNet          & 72.11 & 1.08   & 1.65  \\ 
AG News TextCNN              & 91.03 & 25.66  & 27.58  \\ 
AG News MiniBert             & 90.94 & 25.54  & 25.39  \\ \bottomrule
\end{tabular}
\end{threeparttable}
\end{table}

%% file: tables/time_acc.tex
\begin{table*}[t]
\renewcommand\arraystretch{0.8}
\centering
\caption{Total training time, watermark embed/verify time (ms), and extraction accuracy (\%) across MNIST, CIFAR10/100, and AG News models. We vary 
$\varepsilon$, $\gamma$ and watermark length $B$. Overhead remains small even at $B$=2048; accuracy $>$99\%. Larger $\gamma$ increases total training time; ``Baseline'' has no watermarking.}
\label{time-acc}
\resizebox{\linewidth}{!}{
\begin{tabular}{cccccccccc}
\toprule
                                     & Time/Acc           & \makecell{MNIST \\ DeepLeNet}              & \makecell{MNIST \\ AlexNet}                & \makecell{CIFAR10 \\ AlexNet}               & \makecell{CIFAR10 \\ ResNet18}              & \makecell{CIFAR100 \\ ResNet18}              & \makecell{CIFAR100 \\ WideResNet}            & \makecell{AG News \\ TextCNN}               & \makecell{AG News \\ MiniBert}              \\ 
\midrule
                                     & Training time (s)  & \cellcolor[HTML]{FFFFFF}629.98 & \cellcolor[HTML]{FFFFFF}509.71 & \cellcolor[HTML]{FFFFFF}473.07 & \cellcolor[HTML]{FFFFFF}1013.43 & \cellcolor[HTML]{FFFFFF}1013.63 & \cellcolor[HTML]{FFFFFF}6267.04 & \cellcolor[HTML]{FFFFFF}288.93 & \cellcolor[HTML]{FFFFFF}1488.34 \\
\multirow{-2}{*}{Baseline}           & Watermark Acc (\%)  & \cellcolor[HTML]{FFFFFF}0      & \cellcolor[HTML]{FFFFFF}0      & \cellcolor[HTML]{FFFFFF}0      & \cellcolor[HTML]{FFFFFF}0       & \cellcolor[HTML]{FFFFFF}0       & \cellcolor[HTML]{FFFFFF}0 & \cellcolor[HTML]{FFFFFF}0      & \cellcolor[HTML]{FFFFFF}0       \\ 
\midrule
                                     & Training time (s) & \cellcolor[HTML]{FFFFFF}778.87 & \cellcolor[HTML]{FFFFFF}623.61 & \cellcolor[HTML]{FFFFFF}611.44 & \cellcolor[HTML]{FFFFFF}1200.14 & \cellcolor[HTML]{FFFFFF}1214.14 & \cellcolor[HTML]{FFFFFF}7188.71 & \cellcolor[HTML]{FFFFFF}389.74 & \cellcolor[HTML]{FFFFFF}1614.01 \\
                                     & Embed time (ms)    & \cellcolor[HTML]{FFFFFF}8.68   & \cellcolor[HTML]{FFFFFF}8.94   & \cellcolor[HTML]{FFFFFF}8.83   & \cellcolor[HTML]{FFFFFF}13.17   & \cellcolor[HTML]{FFFFFF}13.85   & \cellcolor[HTML]{FFFFFF}9.56    & \cellcolor[HTML]{FFFFFF}8.31   & \cellcolor[HTML]{FFFFFF}8.17    \\
                                     & Verify time (ms)   & \cellcolor[HTML]{FFFFFF}8.65   & \cellcolor[HTML]{FFFFFF}8.95   & \cellcolor[HTML]{FFFFFF}8.67   & \cellcolor[HTML]{FFFFFF}12.99   & \cellcolor[HTML]{FFFFFF}12.98   & \cellcolor[HTML]{FFFFFF}9.45    & \cellcolor[HTML]{FFFFFF}8.29   & \cellcolor[HTML]{FFFFFF}8.17    \\
\multirow{-4}{*}{$\varepsilon$=2.0}  & Watermark Acc (\%)  & \cellcolor[HTML]{FFFFFF}99.67  & \cellcolor[HTML]{FFFFFF}99.81  & \cellcolor[HTML]{FFFFFF}99.78  & \cellcolor[HTML]{FFFFFF}99.56   & \cellcolor[HTML]{FFFFFF}99.67   & \cellcolor[HTML]{FFFFFF}99.21   & \cellcolor[HTML]{FFFFFF}99.17  & \cellcolor[HTML]{FFFFFF}99.78   \\ 
\cmidrule{9-10}
                                     & Training time (s) & \cellcolor[HTML]{FFFFFF}784.98 & \cellcolor[HTML]{FFFFFF}633.46 & \cellcolor[HTML]{FFFFFF}628.14 & \cellcolor[HTML]{FFFFFF}1192.48 & \cellcolor[HTML]{FFFFFF}1215.37 & \cellcolor[HTML]{FFFFFF}7463.68 & \cellcolor[HTML]{FFFFFF}393.29 & \cellcolor[HTML]{FFFFFF}1650.14 \\
                                     & Embed time (ms)    & \cellcolor[HTML]{FFFFFF}8.12   & \cellcolor[HTML]{FFFFFF}8.72   & \cellcolor[HTML]{FFFFFF}8.66   & \cellcolor[HTML]{FFFFFF}13.31   & \cellcolor[HTML]{FFFFFF}13.69   & \cellcolor[HTML]{FFFFFF}9.19    & \cellcolor[HTML]{FFFFFF}8.53   & \cellcolor[HTML]{FFFFFF}7.97    \\
                                     & Verify time (ms)   & \cellcolor[HTML]{FFFFFF}8.22   & \cellcolor[HTML]{FFFFFF}8.65   & \cellcolor[HTML]{FFFFFF}8.44   & \cellcolor[HTML]{FFFFFF}13.05   & \cellcolor[HTML]{FFFFFF}13.65   & \cellcolor[HTML]{FFFFFF}9.11    & \cellcolor[HTML]{FFFFFF}8.47   & \cellcolor[HTML]{FFFFFF}7.67    \\
\multirow{-4}{*}{$\varepsilon$=5.0}  & Watermark Acc (\%)  & \cellcolor[HTML]{FFFFFF}99.56  & \cellcolor[HTML]{FFFFFF}98.89  & \cellcolor[HTML]{FFFFFF}99.75  & \cellcolor[HTML]{FFFFFF}99.35   & \cellcolor[HTML]{FFFFFF}99.32   & \cellcolor[HTML]{FFFFFF}99.23   & \cellcolor[HTML]{FFFFFF}99.27  & \cellcolor[HTML]{FFFFFF}98.89   \\ 
\cmidrule{3-4}
                                     & Training time (s) & \cellcolor[HTML]{FFFFFF}786.43 & \cellcolor[HTML]{FFFFFF}640.82 & \cellcolor[HTML]{FFFFFF}639.48 & \cellcolor[HTML]{FFFFFF}1199.61 & \cellcolor[HTML]{FFFFFF}1206.17 & \cellcolor[HTML]{FFFFFF}7687.28 & \cellcolor[HTML]{FFFFFF}399.01 & \cellcolor[HTML]{FFFFFF}1576.13 \\
                                     & Embed time (ms)    & \cellcolor[HTML]{FFFFFF}8.61   & \cellcolor[HTML]{FFFFFF}8.77   & \cellcolor[HTML]{FFFFFF}8.81   & \cellcolor[HTML]{FFFFFF}13.41   & \cellcolor[HTML]{FFFFFF}13.64   & \cellcolor[HTML]{FFFFFF}9.27    & \cellcolor[HTML]{FFFFFF}7.51   & \cellcolor[HTML]{FFFFFF}7.47    \\
                                     & Verify time (ms)   & \cellcolor[HTML]{FFFFFF}8.62   & \cellcolor[HTML]{FFFFFF}8.57   & \cellcolor[HTML]{FFFFFF}8.81   & \cellcolor[HTML]{FFFFFF}13.34   & \cellcolor[HTML]{FFFFFF}13.56   & \cellcolor[HTML]{FFFFFF}9.21    & \cellcolor[HTML]{FFFFFF}7.47   & \cellcolor[HTML]{FFFFFF}7.35    \\
\multirow{-4}{*}{$\varepsilon$=10.0} & Watermark Acc (\%)  & \cellcolor[HTML]{FFFFFF}98.89  & \cellcolor[HTML]{FFFFFF}99.21  & \cellcolor[HTML]{FFFFFF}99.37  & \cellcolor[HTML]{FFFFFF}98.99   & \cellcolor[HTML]{FFFFFF}99.24   & \cellcolor[HTML]{FFFFFF}99.24   & \cellcolor[HTML]{FFFFFF}99.18  & \cellcolor[HTML]{FFFFFF}99.03   \\ 
\midrule
                                     & Training time (s) & \cellcolor[HTML]{FFFFFF}786.43 & \cellcolor[HTML]{FFFFFF}640.82 & \cellcolor[HTML]{FFFFFF}639.48 & \cellcolor[HTML]{FFFFFF}1199.61 & \cellcolor[HTML]{FFFFFF}1206.17 & \cellcolor[HTML]{FFFFFF}7687.28 & \cellcolor[HTML]{FFFFFF}399.01 & \cellcolor[HTML]{FFFFFF}1576.13 \\
                                     & Embed time (ms)    & \cellcolor[HTML]{FFFFFF}8.61   & \cellcolor[HTML]{FFFFFF}8.77   & \cellcolor[HTML]{FFFFFF}8.81   & \cellcolor[HTML]{FFFFFF}13.41   & \cellcolor[HTML]{FFFFFF}13.64   & \cellcolor[HTML]{FFFFFF}9.27    & \cellcolor[HTML]{FFFFFF}7.51   & \cellcolor[HTML]{FFFFFF}7.47    \\
                                     & Verify time (ms)   & \cellcolor[HTML]{FFFFFF}8.62   & \cellcolor[HTML]{FFFFFF}8.57   & \cellcolor[HTML]{FFFFFF}8.81   & \cellcolor[HTML]{FFFFFF}13.34   & \cellcolor[HTML]{FFFFFF}13.56   & \cellcolor[HTML]{FFFFFF}9.21    & \cellcolor[HTML]{FFFFFF}7.47   & \cellcolor[HTML]{FFFFFF}7.35    \\
\multirow{-4}{*}{$\gamma$=1.5}       & Watermark Acc (\%)  & \cellcolor[HTML]{FFFFFF}99.32  & \cellcolor[HTML]{FFFFFF}99.87  & \cellcolor[HTML]{FFFFFF}99.37  & \cellcolor[HTML]{FFFFFF}99.21   & \cellcolor[HTML]{FFFFFF}99.18   & \cellcolor[HTML]{FFFFFF}99.23   & \cellcolor[HTML]{FFFFFF}99.28  & \cellcolor[HTML]{FFFFFF}99.38   \\ 
\cmidrule{5-6}
                                     & Training time (s) & \cellcolor[HTML]{FFFFFF}918.69 & \cellcolor[HTML]{FFFFFF}809.33 & \cellcolor[HTML]{FFFFFF}802.27 & \cellcolor[HTML]{FFFFFF}1504.91 & \cellcolor[HTML]{FFFFFF}1517.11 & \cellcolor[HTML]{FFFFFF}9627.41 & \cellcolor[HTML]{FFFFFF}511.65 & \cellcolor[HTML]{FFFFFF}2008.91 \\
                                     & Embed time (ms)    & \cellcolor[HTML]{FFFFFF}8.61   & \cellcolor[HTML]{FFFFFF}8.52   & \cellcolor[HTML]{FFFFFF}9.21   & \cellcolor[HTML]{FFFFFF}13.37   & \cellcolor[HTML]{FFFFFF}13.62   & \cellcolor[HTML]{FFFFFF}8.98    & \cellcolor[HTML]{FFFFFF}7.99   & \cellcolor[HTML]{FFFFFF}7.32    \\
                                     & Verify time (ms)   & \cellcolor[HTML]{FFFFFF}8.61   & \cellcolor[HTML]{FFFFFF}8.67   & \cellcolor[HTML]{FFFFFF}9.23   & \cellcolor[HTML]{FFFFFF}13.28   & \cellcolor[HTML]{FFFFFF}13.34   & \cellcolor[HTML]{FFFFFF}9.01    & \cellcolor[HTML]{FFFFFF}7.89   & \cellcolor[HTML]{FFFFFF}7.25    \\
\multirow{-4}{*}{$\gamma$=2.0}       & Watermark Acc (\%)  & \cellcolor[HTML]{FFFFFF}99.45  & \cellcolor[HTML]{FFFFFF}99.78  & \cellcolor[HTML]{FFFFFF}99.38  & \cellcolor[HTML]{FFFFFF}99.25   & \cellcolor[HTML]{FFFFFF}99.52   & \cellcolor[HTML]{FFFFFF}99.21   & \cellcolor[HTML]{FFFFFF}99.28  & \cellcolor[HTML]{FFFFFF}99.26   \\ 
\cmidrule{7-8}
                                     & Training time (s) & \cellcolor[HTML]{FFFFFF}989.13 & \cellcolor[HTML]{FFFFFF}817.74 & \cellcolor[HTML]{FFFFFF}819.56 & \cellcolor[HTML]{FFFFFF}1538.17 & \cellcolor[HTML]{FFFFFF}1548.11 & \cellcolor[HTML]{FFFFFF}9691.42 & \cellcolor[HTML]{FFFFFF}511.15 & \cellcolor[HTML]{FFFFFF}2080.57 \\
                                     & Embed time (ms)    & \cellcolor[HTML]{FFFFFF}8.53   & \cellcolor[HTML]{FFFFFF}8.68   & \cellcolor[HTML]{FFFFFF}9.34   & \cellcolor[HTML]{FFFFFF}13.99   & \cellcolor[HTML]{FFFFFF}13.53   & \cellcolor[HTML]{FFFFFF}9.01    & \cellcolor[HTML]{FFFFFF}7.42   & \cellcolor[HTML]{FFFFFF}7.09    \\
                                     & Verify time (ms)   & \cellcolor[HTML]{FFFFFF}8.51   & \cellcolor[HTML]{FFFFFF}8.63   & \cellcolor[HTML]{FFFFFF}9.12   & \cellcolor[HTML]{FFFFFF}12.98   & \cellcolor[HTML]{FFFFFF}13.22   & \cellcolor[HTML]{FFFFFF}8.99    & \cellcolor[HTML]{FFFFFF}7.56   & \cellcolor[HTML]{FFFFFF}7.21    \\
\multirow{-4}{*}{$\gamma$=2.5}       & Watermark Acc (\%)  & \cellcolor[HTML]{FFFFFF}98.99  & \cellcolor[HTML]{FFFFFF}99.43  & \cellcolor[HTML]{FFFFFF}98.87  & \cellcolor[HTML]{FFFFFF}99.26   & \cellcolor[HTML]{FFFFFF}99.63   & \cellcolor[HTML]{FFFFFF}99.23   & \cellcolor[HTML]{FFFFFF}99.18  & \cellcolor[HTML]{FFFFFF}99.24   \\ 
\midrule
                                     & Training time (s) & \cellcolor[HTML]{FFFFFF}786.43 & \cellcolor[HTML]{FFFFFF}640.82 & \cellcolor[HTML]{FFFFFF}639.48 & \cellcolor[HTML]{FFFFFF}1199.61 & \cellcolor[HTML]{FFFFFF}1206.17 & \cellcolor[HTML]{FFFFFF}7687.28 & \cellcolor[HTML]{FFFFFF}399.01 & \cellcolor[HTML]{FFFFFF}1576.13 \\
                                     & Embed time (ms)    & \cellcolor[HTML]{FFFFFF}8.61   & \cellcolor[HTML]{FFFFFF}8.77   & \cellcolor[HTML]{FFFFFF}8.81   & \cellcolor[HTML]{FFFFFF}13.41   & \cellcolor[HTML]{FFFFFF}13.64   & \cellcolor[HTML]{FFFFFF}9.27    & \cellcolor[HTML]{FFFFFF}7.51   & \cellcolor[HTML]{FFFFFF}7.47    \\
                                     & Verify time (ms)   & \cellcolor[HTML]{FFFFFF}8.62   & \cellcolor[HTML]{FFFFFF}8.57   & \cellcolor[HTML]{FFFFFF}8.81   & \cellcolor[HTML]{FFFFFF}13.34   & \cellcolor[HTML]{FFFFFF}13.56   & \cellcolor[HTML]{FFFFFF}9.21    & \cellcolor[HTML]{FFFFFF}7.47   & \cellcolor[HTML]{FFFFFF}7.35    \\
\multirow{-4}{*}{$B$=512}            & Watermark Acc (\%)  & \cellcolor[HTML]{FFFFFF}99.21  & \cellcolor[HTML]{FFFFFF}99.48  & \cellcolor[HTML]{FFFFFF}99.89  & \cellcolor[HTML]{FFFFFF}99.26   & \cellcolor[HTML]{FFFFFF}99.02   & \cellcolor[HTML]{FFFFFF}98.89   & \cellcolor[HTML]{FFFFFF}99.25  & \cellcolor[HTML]{FFFFFF}99.34   \\ 
\cmidrule{3-4}
                                     & Training time (s) & \cellcolor[HTML]{FFFFFF}803.34 & \cellcolor[HTML]{FFFFFF}659.94 & \cellcolor[HTML]{FFFFFF}661.51 & \cellcolor[HTML]{FFFFFF}1185.81 & \cellcolor[HTML]{FFFFFF}1206.68 & \cellcolor[HTML]{FFFFFF}7654.09 & \cellcolor[HTML]{FFFFFF}416.78 & \cellcolor[HTML]{FFFFFF}1637.21 \\
                                     & Embed time (ms)    & \cellcolor[HTML]{FFFFFF}12.89  & \cellcolor[HTML]{FFFFFF}12.88  & \cellcolor[HTML]{FFFFFF}14.23  & \cellcolor[HTML]{FFFFFF}22.64   & \cellcolor[HTML]{FFFFFF}21.53   & \cellcolor[HTML]{FFFFFF}13.42   & \cellcolor[HTML]{FFFFFF}12.45  & \cellcolor[HTML]{FFFFFF}13.13   \\
                                     & Verify time (ms)   & \cellcolor[HTML]{FFFFFF}12.83  & \cellcolor[HTML]{FFFFFF}12.65  & \cellcolor[HTML]{FFFFFF}14.32  & \cellcolor[HTML]{FFFFFF}22.46   & \cellcolor[HTML]{FFFFFF}21.78   & \cellcolor[HTML]{FFFFFF}13.21   & \cellcolor[HTML]{FFFFFF}13.01  & \cellcolor[HTML]{FFFFFF}13.27   \\
\multirow{-4}{*}{$B$=1024}           & Watermark Acc (\%)  & \cellcolor[HTML]{FFFFFF}99.36  & \cellcolor[HTML]{FFFFFF}99.21  & \cellcolor[HTML]{FFFFFF}99.21  & \cellcolor[HTML]{FFFFFF}99.37   & \cellcolor[HTML]{FFFFFF}99.71   & \cellcolor[HTML]{FFFFFF}99.24   & \cellcolor[HTML]{FFFFFF}99.39  & \cellcolor[HTML]{FFFFFF}99.28   \\ 
\cmidrule{7-8}
                                     & Training time (s) & \cellcolor[HTML]{FFFFFF}906.76 & \cellcolor[HTML]{FFFFFF}680.91 & \cellcolor[HTML]{FFFFFF}680.39 & \cellcolor[HTML]{FFFFFF}1225.16 & \cellcolor[HTML]{FFFFFF}1247.42 & \cellcolor[HTML]{FFFFFF}8013.85 & \cellcolor[HTML]{FFFFFF}433.59 & \cellcolor[HTML]{FFFFFF}1773.24 \\
                                     & Embed time (ms)    & \cellcolor[HTML]{FFFFFF}21.41  & \cellcolor[HTML]{FFFFFF}25.57  & \cellcolor[HTML]{FFFFFF}25.98  & \cellcolor[HTML]{FFFFFF}43.79   & \cellcolor[HTML]{FFFFFF}40.65   & \cellcolor[HTML]{FFFFFF}24.48   & \cellcolor[HTML]{FFFFFF}24.89  & \cellcolor[HTML]{FFFFFF}24.21   \\
                                     & Verify time (ms)   & \cellcolor[HTML]{FFFFFF}21.22  & \cellcolor[HTML]{FFFFFF}25.22  & \cellcolor[HTML]{FFFFFF}24.99  & \cellcolor[HTML]{FFFFFF}43.78   & \cellcolor[HTML]{FFFFFF}40.22   & \cellcolor[HTML]{FFFFFF}24.68   & \cellcolor[HTML]{FFFFFF}24.62  & \cellcolor[HTML]{FFFFFF}24.18   \\
\multirow{-4}{*}{$B$=2048}           & Watermark Acc (\%)  & \cellcolor[HTML]{FFFFFF}99.68  & \cellcolor[HTML]{FFFFFF}99.23  & \cellcolor[HTML]{FFFFFF}99.79  & \cellcolor[HTML]{FFFFFF}99.29   & \cellcolor[HTML]{FFFFFF}98.87   & \cellcolor[HTML]{FFFFFF}99.23   & \cellcolor[HTML]{FFFFFF}99.09  & \cellcolor[HTML]{FFFFFF}99.13   \\ 
\bottomrule
\end{tabular}
}

\end{table*}

%% file: tables/communication-T.tex
\begin{table}[]
\centering
\caption{
Per-link communication latency (s) between adjacent parties.}
\label{table-communicatoin}
\renewcommand{\arraystretch}{1.0}
\resizebox{\columnwidth}{!}{
\begin{threeparttable}
\begin{tabular}{lcccc}
\toprule
\multicolumn{1}{c}{Transfer} & $ \mathcal{C} \longrightarrow \mathcal{T}_1 $ & $ \mathcal{T}_1 \longrightarrow \mathcal{T}_2 $ & $\mathcal{T}_2 \longrightarrow \mathcal{T}_3 $ & $T_{total}$ \\ 
\midrule
MNIST DeepLeNet             & 0.24 & 0.12  & 0.06  & 0.42 \\ 
MNIST AlexNet               & 0.06 & 0.09  & 0.01  & 0.16 \\ 
CIFAR10 AlexNet              & 0.08 & 0.12  & 0.02  & 0.22 \\ 
CIFAR10 ResNet18             & 0.32 & 0.16  & 0.08  & 0.56 \\ 
CIFAR100 ResNet18            & 0.32 & 0.16  & 0.08  & 0.56 \\ 
CIFAR100 WideResNet          & 0.80 & 0.40  & 0.20  & 1.40 \\ 
AG News TextCNN              & 0.08 & 0.01  & 0.01  & 0.10 \\ 
AG News MiniBert             & 0.08 & 0.06  & 0.06  & 0.20 \\ \bottomrule
\end{tabular}
\end{threeparttable}
}
\end{table}

%% file: sections_1/6_relatedwork.tex
\section{Related Works}\label{sec:related work}


\begin{table}[t]
\centering
\caption{
Existing privacy protection methods for SL}
\label{tab_label_privacy}
\renewcommand{\arraystretch}{1.1} 
\resizebox{\linewidth}{!}{
\begin{tabular}{c|cccc}
\toprule

 \diagbox{\textbf{Method}}{\textbf{Protection}}  & \textbf{\makecell{Label \\ Privacy}} & \textbf{\makecell{Activation \\ Data}} & \textbf{\makecell{Performance \\ Preserved}} & \textbf{\makecell{Copyright \& \\ Train Trace}}  \\
\midrule
Vanilla SL~\cite{vepakomma2018split_13, poirot2019split_40}   &   \xmark &  \xmark & \cmark  & \xmark \\
U-shaped SL~\cite{gupta2018distributed_12, u_shaped_38, u_shaped_39}   &   \cmark &  \xmark & \cmark  & \xmark \\
B-SL~\cite{B_SL_46}   &   \xmark &  \cmark & \xmark  & \xmark \\
R3eLU~\cite{secure_SL_42}   &  \xmark &  \cmark & \xmark  & \xmark \\
\midrule

\textbf{Ours}   &   \cmark &  \cmark &  \cmark  & \cmark \\

\bottomrule

\end{tabular}
}
\end{table}

\subsubsection{SL Designs}

Users are increasingly unwilling to disclose sensitive information. SL~\cite{matsubara2022split_6} addresses this by enabling outsourced training without sharing raw data. 
SplitNN~\cite{vepakomma2018split_13} demonstrated SL for medical diagnostics, enabling collaboration between centralized,
and local institutions while protecting patient privacy. Poirot et al.~\cite{poirot2019split_40} extended SL to multi-center healthcare settings with restricted data sharing. Hu~\cite{hu2025split_7} highlighted the temporal nature of clinical data and proposed grouping patient samples by time to better capture sequential dependencies in SL.
Gupta et al.~\cite{gupta2018distributed_12} placed the output layer on the client to prevent label exposure. Lyu et al.~\cite{u_shaped_39} added parallel computing and optimized resource allocation to improve edge efficiency. UVSL~\cite{u_shaped_38} combined the U-shaped design with DP to protect both labels and data.

\subsubsection{SL Privacy Threats and Defenses}

SL is a privacy-preserving alternative to centralized training, but remains vulnerable to label and feature inference attacks.
Abuadbba et al.~\cite{abuadbba2020can_15} showed that 1D CNNs trained with SL can match non-split accuracy yet still leak privacy, as servers can infer input features from intermediate outputs, constituting the first feature inference attack in SL. RecoNN~\cite{inversion_26} trains a reconstruction network using shadow models and shadow samples; at attack time, target weights are input to reconstruct images of the target’s training data. Unsplit~\cite{unsplit_27} jointly searches the input and parameter spaces to reconstruct inputs and produce a functional clone whose outputs approximate the client’s smashed data. FORA~\cite{xu2024_FORA_41} uses public auxiliary data for feature-level transfer, building a surrogate client to train an attack model that reconstructs private inputs covertly. Pasquini et al.~\cite{pasquini2021unleashing_16} verified across settings that a malicious server can reconstruct a client’s private dataset via feature inference.

To mitigate these threats, several defenses have been explored. UVSL~\cite{u_shaped_38} keeps label-sensitive layers on the client and adds local DP noise to smashed activations, limiting label inference and input reconstruction. Lyu et al.~\cite{u_shaped_39} placed the model tail on the client and perform local backpropagation, preserving label confidentiality. U-shaped designs reduce label exposure but increase communication, since each iteration requires at least two round-trips between client and server.
Ngoc et al.~\cite{B_SL_46} proposed Binarized SL (B-SL), which binarizes local layers and applies differential privacy to mixed smashed data. Mao et al.~\cite{secure_SL_42} introduce a randomized activation, R3eLU, to perturb forward smashed data and backward local gradients while integrating a DP mechanism. 

\subsubsection{Copyright Protection}
As DNNs are co-developed across institutions, model copyright~\cite{AIcopyright_74} and training traceability are essential. Copyright protection secures legitimate ownership and proper attribution for contributors. Training traceability verifies that trainers performed the required and correct procedures, ensuring accountability and trust in collaborative settings~\cite{toward-FL}.
EWDNN~\cite{uchida2017embedding_18} embeds binary watermarks into weights via a training regularizer, enabling white-box extraction without accuracy loss. RIGA~\cite{wang2021riga_19} adds secret keys, nonlinear mappings, and redundancy to resist detection and removal. Both methods assume centralized, single-node training and are unsuitable for collaborative settings. For distributed scenarios, FedIPR~\cite{li2022fedipr_20} verifies client ownership in federated learning by embedding user-specific watermarks during local training and using challenge–response verification. FedIPR targets gradient-synchronized FL and does not transfer to SL, where model segments are asymmetrically partitioned. 


\textbf{Ours!} Those solutions are not perfect (as compared in Table~\ref{tab_label_privacy}). We introduce label expansion that obscures both label counts and semantics, preventing the trainers from inferring true labels or the task, and enabling legitimate use only by parties holding the true$\leftrightarrow$pseudo mapping. We also use chained watermarking that cryptographically links successive training stages to ensure verifiable training integrity. \textsc{CliCooper} thus allows multiple compute-constrained clients to train the model and removes reliance on a single high-capacity server.

%% file: sections_1/7_conclusion.tex
\section{Concluding Remarks}
\label{sec:conclusion}

We presented \textsc{CliCooper}, a new multi-client SL framework that collaboratively preserves data privacy, enables verifiable training, and protects model copyright. Our method integrates label expansion to mask true labels and prevent unauthorized model use, and a Gaussian noise-based mechanism to safeguard smashed data from inversion. In addition, dynamic chained watermarks provide cryptographically linked training traceability and copyright protection.
We experimentally confirmed that \textsc{CliCooper} maintains 100\% of baseline accuracy, with some models even achieving up to 2\% accuracy. It effectively defends against strong adversaries: the success rate of internal clustering attacks reduces to 0\%, the similarity of inversion-reconstructed samples drops from 45\% to 0.03\%, and the surrogate model accuracy from model extraction attacks falls to around 1\%, equivalent to random guessing.


\input{tables/table-notation}

%% file: tables/table-notation.tex
\begin{table}[!hbt]
    \renewcommand\arraystretch{1.0}
    \small
    \caption{List of Notation}
    \label{tab:notation}
    \begin{center}
    \begin{threeparttable}
    \resizebox{\linewidth}{!}{
    \begin{tabular}{cl}
    \toprule
    \multicolumn{1}{c}{\textbf{Notation}} &  \multicolumn{1}{c}{\textbf{Definition}} \\
    \midrule
   \multicolumn{1}{c}{\cellcolor{yellow!15}$\mathcal{C}$/$\mathcal{T}$/$\mathcal{V}$} & \multicolumn{1}{l}{Data client/Trainer client/Verifier} \\

     \multicolumn{1}{c}{\cellcolor{yellow!15}$\mathbb{D}^*$/ $\mathbb{D}$/ $\mathbb{D}_{t}$} & \multicolumn{1}{l}{Final/Original/Test Dataset} \\
    \cellcolor{yellow!15}$Y$/$Y^*$ & \multicolumn{1}{l}{True/Pseudo label of Dataset} \\
    \cellcolor{yellow!15}$q$ & \multicolumn{1}{l}{The class number of the dataset} \\
    \cellcolor{yellow!15}$g_i$ & \multicolumn{1}{l}{The expansion factor of the labels $Y_{i}^{*}$} \\
    \cellcolor{yellow!15}$\gamma$ & \multicolumn{1}{l}{The average expansion factor of $Y^*$ relative to $Y$} \\
    \cellcolor{yellow!15}$\mathcal{G}_Y$/$\mathcal{G}_{Y}^{-1}$ & \multicolumn{1}{l}{The one-to-many/many-to-one mapping between true and pseudo label} \\
    \cellcolor{yellow!15}$\varepsilon$/$\mathbb{G(\cdot )}$ & \multicolumn{1}{l}{DP privacy budget/DP noise injection} \\

    \cellcolor{yellow!15}$\!\mathbb{M}_{\mathcal{T}_{i}}$ & \multicolumn{1}{l}{$\mathcal{T}_{i}$'s activation for $\mathcal{T}_{i+1}$} \\
    \cellcolor{blue!8}$\mathbb{M}_{\mathcal{C}}^{\text{\tiny DP}}$ & \multicolumn{1}{l}{DP-protected activation} \\
    \cellcolor{blue!8}$\bar{A}$ & \multicolumn{1}{l}{the activation of the $\mathcal{C}$} \\
    \cellcolor{blue!8}$W$ & \multicolumn{1}{l}{The whole SL model/weights} \\
    \cellcolor{blue!8}$W_{\mathcal{T}_{i}}$ & \multicolumn{1}{l}{The $i$th trainer's model/weights} \\
    \cellcolor{blue!8}$W_{\mathcal{C}}$ & \multicolumn{1}{l}{$\mathcal{C}$'s pre-trained encoder/model} \\
    \cellcolor{blue!8}$\Lambda$ & \multicolumn{1}{l}{Copyright information embedded in the model} \\
    \cellcolor{blue!8}$\Lambda_{\mathcal{T}_{i}}$ & \multicolumn{1}{l}{$\mathcal{T}_{i}$'s watermark} \\
    \cellcolor{yellow!15}$B$ & \multicolumn{1}{l}{The watermark size} \\
    \multicolumn{1}{c}
    {\cellcolor{yellow!15}$k_{\mathcal{T}_{i}}$} & Embedding key for $\Lambda_{\mathcal{T}_{i}}$ \\
    \multicolumn{1}{c}{\cellcolor{yellow!15}$Z_{\mathcal{T}_{i}}$} & Selection matrix for weights in $W_{\mathcal{T}_{i}}$, used for $\Lambda_{\mathcal{T}_{i}}$ \\
    \multicolumn{1}{c}{\cellcolor{yellow!15}$\mu_i$} & Unique secret nonce used in hash computation \\
    \multicolumn{1}{c}{\cellcolor{yellow!15}$ID_{\mathcal{T}_{i}}$} & Unique $\mathcal{T}_{i}$ node identity\\
    \multicolumn{1}{c}{\cellcolor{yellow!15}$\eta$} & Watermark detection rate \\
    \multicolumn{1}{c}{\cellcolor{blue!8}$\mathbb{O}(\cdot)$} & \multicolumn{1}{l}{Watermark extraction function} \\
    \multicolumn{1}{c}{\cellcolor{blue!8}$l_w(\cdot)$/$l_\Lambda(\cdot)$} & \multicolumn{1}{l}{The loss function for main task and embedding watermark} \\
    \bottomrule
    \end{tabular}
    }
    \end{threeparttable}
    \end{center}
    \vspace{-0.1in}
    \end{table}